%% file: main.tex
\documentclass{article}

\PassOptionsToPackage{numbers}{natbib}
\usepackage[preprint]{neurips_2023}

\usepackage[utf8]{inputenc} % allow utf-8 input
\usepackage[T1]{fontenc}    % use 8-bit T1 fonts
\usepackage{hyperref}       % hyperlinks
\usepackage{url}            % simple URL typesetting
\usepackage{booktabs}       % professional-quality tables
\usepackage{amsfonts}       % blackboard math symbols
\usepackage{nicefrac}       % compact symbols for 1/2, etc.
\usepackage{microtype}      % microtypography
\usepackage{xcolor}         % colors
\usepackage[algoruled,boxed,noline,noend]{algorithm2e}
\usepackage{amsmath}
\usepackage{graphicx}
\usepackage{siunitx}
\usepackage{wrapfig}
\usepackage{amsthm}
\usepackage{bm}
\usepackage{tikz}
\usepackage{etoolbox}

\newtheorem{theorem}{Theorem}

\newtheorem{claim}[theorem]{Claim}

\DeclareSIUnit{\million}{\text{million}}
\DeclareMathAlphabet\mathbfcal{OMS}{cmsy}{b}{n}

\SetEndCharOfAlgoLine{}  % Don't end algorithm lines with ;.

\newcounter{mosNoteCounter}

\newcommand{\newterm}[1]{\textbf{\textit{#1}}}

\newrobustcmd{\profile}[1]{%
\raisebox{-.4 ex}{
\begin{tikzpicture}[scale=0.08]
    \draw[gray,step=1cm] (0,0) grid (4,4); % inner grid lines in gray
    \draw[black] (0,0) rectangle (4,4); % outer grid lines in black
    
    \foreach \x/\y in {#1} {
        \fill[black] (\x, \y) rectangle ++(1,1);
    }
\end{tikzpicture}
}
}

\newrobustcmd{\wmA}{\profile{0/3}}
\newrobustcmd{\wmB}{\profile{1/2}}
\newrobustcmd{\wmC}{\profile{2/1}}
\newrobustcmd{\wmAB}{\profile{0/3,1/2,1/3,0/2}}
\newrobustcmd{\wmBC}{\profile{1/1,2/2,1/2,2/1}}
\newrobustcmd{\wmAC}{\profile{0/3,2/1,2/3,0/1}}
\newrobustcmd{\wmABC}{\profile{1/1,2/2,0/1,0/2,1/2,2/1,0/3,1/3,2/3}}

\title{Co-Learning Empirical Games and World Models}

\author{%
Max Olan Smith\\
University of Michigan\\
\texttt{mxsmith@umich.edu} \\
\And
Michael P. Wellman \\
University of Michigan \\
\texttt{wellman@umich.edu} \\
}

\begin{document}

\maketitle

\begin{abstract}
Game-based decision-making involves reasoning over both world dynamics and strategic interactions among the agents.
Typically, empirical models capturing these respective aspects are learned and used separately.
We investigate the potential gain from co-learning these elements: a world model for dynamics and an empirical game for strategic interactions.
Empirical games drive world models toward a broader consideration of possible game dynamics induced by a diversity of strategy profiles. 
Conversely, world models guide empirical games to efficiently discover new strategies through planning. 
We demonstrate these benefits first independently, then in combination as realized by a new algorithm, Dyna-PSRO, that co-learns an empirical game and a world model.
When compared to PSRO---a baseline empirical-game building algorithm, Dyna-PSRO is found to compute lower regret solutions on partially observable general-sum games.
In our experiments, Dyna-PSRO also requires substantially fewer experiences than PSRO, a key algorithmic advantage for settings where collecting player-game interaction data is a cost-limiting factor.
\end{abstract}

\input{body.tex}

\nocite{*}
\bibliographystyle{plain}
\bibliography{bibliography}

\newpage
\appendix
\input{appendix.tex}

\end{document}

%% file: body.tex
\section{Introduction}
\label{sec:introduction}
Even seemingly simple games can actually embody a level of complexity rendering them intractable to direct reasoning.
This complexity stems from the interplay of two sources: dynamics of the game environment, and strategic interactions among the game's players.
As an alternative to direct reasoning, models have been developed to facilitate reasoning over these distinct aspects of the game.
\newterm{Empirical games} capture strategic interactions in the form of payoff estimates for joint policies~\cite{wellman06egta}.
\newterm{World models} represent a game's transition dynamics and reward signal directly~\cite{sutton18book,ha18worldmodels}.
Whereas each of these forms of model have been found useful for game reasoning, typical use in prior work has focused on one or the other, learned and employed in isolation from its natural counterpart.

Co-learning both models presents an opportunity to leverage their complementary strengths as a means to improve each other.
World models predict successor states and rewards given a game's current state and action(s).
However, their performance depends on coverage of their training data, which is limited by the range of strategies considered during learning.
Empirical games can inform training of world models by suggesting a diverse set of salient strategies, based on game-theoretic reasoning~\cite{wellman06egta}.
These strategies can expose the world model to a broader range of relevant dynamics.
Moreover, as empirical games are estimated through simulation of strategy profiles, this same simulation data can be reused as training data for the world model. 

Strategic diversity through empirical games, however, comes at a cost.
In the popular framework of Policy-Space Response Oracles (PSRO)~\cite{lanctot17psro}, empirical normal-form game models are built iteratively, at each step expanding a restricted strategy set by computing best-response policies to the current game's solution.
As computing an exact best-response is generally intractable, PSRO uses Deep Reinforcement Learning (DRL) to compute approximate response policies.
However, each application of DRL can be considerably resource-intensive, necessitating the generation of a vast amount of gameplays for learning. 
Whether gameplays, or experiences, are generated via simulation~\cite{obando2021revisitingrainbow} or from real-world interactions~\cite{hester2012texplore}, their collection poses a major limiting factor in DRL and by extension PSRO.
World models present one avenue to reduce this cost by transferring previously learned game dynamics across response computations.

\begin{wrapfigure}{r}{0.5\textwidth}
  \begin{center}
    \includegraphics[width=0.5\textwidth]{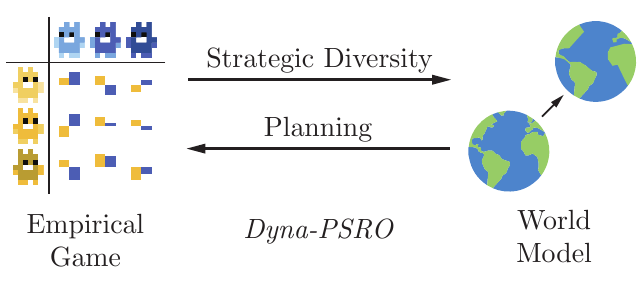}
  \end{center}
  \caption{%
      Dyna-PSRO co-learns a world model and empirical game.
      Empirical games offer world models strategically diverse game dynamics.
      World models offer empirical games more efficient strategy discovery through planning.
  } 
  \label{fig:teaser}
\end{wrapfigure}

We investigate the mutual benefits of co-learning a world model and an empirical game by first verifying the potential contributions of each component independently.
We then show how to realize the combined effects in a new algorithm, \emph{Dyna-PSRO}, that co-learns a world model and an empirical game (illustrated in Figure~\ref{fig:teaser}).
Dyna-PSRO extends PSRO to learn a world model concurrently with empirical game expansion, and applies this world model to reduce the computational cost of computing new policies.
This is implemented by a Dyna-based reinforcement learner~\cite{sutton90dyna, sutton91dyna} that integrates planning, acting, and learning in parallel.
Dyna-PSRO is evaluated against PSRO on a collection of partially observable general-sum games. 
In our experiments, Dyna-PSRO found lower-regret solutions while requiring substantially fewer cumulative experiences.

The main points of novelty of this paper are as follows:
(1)~empirically demonstrate that world models benefit from the strategic diversity induced by an empirical game;
(2)~empirically demonstrate that a world model can be effectively transferred and used in planning with new other-players. 
The major contribution of this work is a new algorithm, Dyna-PSRO, that co-learns an empirical game and world model finding a stronger solution at less cost than the baseline, PSRO.

\section{Related Work}
\label{sec:related-work}

\paragraph{Empirical Game Theoretic Analysis (EGTA).}
The core idea of EGTA~\cite{wellman06egta} is to reason over approximate game models (\textit{empirical games}) estimated by simulation over a restricted strategy set. 
This basic approach was first demonstrated by Walsh et al.~\cite{walsh02egta}, in a study of pricing and bidding games.
Phelps et al.~\cite{phelps06automatic} introduced the idea of extending a strategy set automatically through optimization, employing genetic search over a policy space.
Schvartzman \& Wellman~\cite{Schvartzman09} proposed using RL to derive new strategies that are approximate best responses (BRs) to the current empirical game's Nash equilibrium.
The general question of which strategies to add to an empirical game has been termed the \emph{strategy exploration problem}~\cite{Jordan10sw}.
PSRO~\cite{lanctot17psro} generalized the target for BR beyond NE, and introduced DRL for BR computation in empirical games.
Many further variants and extensions of EGTA have been proposed, for example those using structured game representations such as extensive-form~\cite{mcaleer21xdo, konicki22tree}.
Some prior work has considered transfer learning across BR computations in EGTA, specifically by reusing elements of policies and value functions~\cite{smith20qmixing, smith21}.

\paragraph{Model-Based Reinforcement Learning (MBRL).} 
\emph{Model-Based} RL algorithms construct or use a model of the environment (henceforth, \textit{world model}) in the process of learning a policy or value function~\cite{sutton18book}.
World models may either predict successor observations directly (e.g., at pixel level~\cite{wahlstrom15, watter15}), or in a learned latent space~\cite{ha18recurrentworldmodels, gelada19deepmdp}.
The world models can be either used for \emph{background planning} by rolling out model-predicted trajectories to train a policy, or by \emph{decision-time planning} where the world model is used to evaluate the current state by planning into the future.
Talvitie~\cite{talvitie14} demonstrated that even in small Markov decision processes (MDP)~\cite{puterman94mdp}, model-prediction errors tend to compound---rendering long-term planning at the abstraction of observations ineffective.
A follow-up study demonstrated that for imperfect models, short-term planning was no better than repeatedly training on previously collected real experiences; however, medium-term planning offered advantages even with an imperfect model~\cite{holland18}.
Parallel studies hypothesized that these errors are a result of insufficient data for that transition to be learned~\cite{kurutach18ensemble, buckman18ensemble}.
To remedy the data insufficiency, ensembles of world models were proposed to account for world model uncertainty~\cite{buckman18ensemble, kurutach18ensemble, yu20}, and another line of inquiry used world model uncertainty to guide exploration in state-action space~\cite{ball20, sekar20}.
This study extends this problem into the multiagent setting, where now other-agents may preclude transitions from occurring.
The proposed remedy is to leverage the strategy exploration process of building an empirical game to guide data generation.

\paragraph{Multiagent Reinforcement Learning (MARL).}
Previous research intersecting MARL and MBRL has primarily focused on modeling the opponent, particularly in scenarios where the opponent is fixed and well-defined.
Within specific game sub-classes, like cooperative games and two-player zero-sum games, it has been theoretically shown that opponent modeling reduces the sample complexity of RL~\cite{tian19, zhang20modelcomplexity}.
Opponent models can either explicitly~\cite{mealing15pokermodelling, foerster18lola} or implicitly~\cite{bard13, indarjo19} model the behavior of the opponent.
Additionally, these models can either construct a single model of opponent behavior, or learn a set of models~\cite{collins07, he16opponentmodeling}.
While opponent modeling details are beyond the scope of this study, readers can refer to Albrecht \& Stone's survey~\cite{albrecht18survey} for a comprehensive review on this subject.
Instead, we consider the case where the learner has explicit access to the opponent's policy during training, as is the case in empirical-game building.
A natural example is that of Self-Play, where all agents play the same policy; therefore, a world model can be learned used to evaluate the quality of actions with Monte-Carlo Tree Search~\cite{silver16alphago, silver2017alphazero, tesauro1995tdgammon, schrittwieser20muzero}.
Li et al.~\cite{zun23search} expands on this by building a population of candidate opponent policies through PSRO to augment the search procedure.
Krupnik et al.~\cite{krupnik20marlgenerative} demonstrated that a generative world model could be useful in multi-step opponent-action prediction.
Sun et al.~\cite{sun19} examined modeling stateful game dynamics from observations when the agents' policies are stationary.
Chockalingam et al.~\cite{chockalingam18} explored learning world models for homogeneous agents with a centralized controller in a cooperative game.
World models may also be shared by independent reinforcement learners in cooperative games~\cite{willemsen21mambpo, zhang22marco}.

\section{Co-Learning Benefits}
\label{sec:experiments}
We begin by specifying exactly what we mean by world model and empirical game.
This requires defining some primitive elements.
Let $t\in\mathcal{T}$ denote time in the real game, with $s^t\in\mathcal{S}$ the \newterm{information state} and $h^t\in\mathcal{H}$ the \newterm{game state} at time $t$.
The information state $s^t\equiv(m^{\pi,t}, o^t)$ is composed of the \newterm{agent's memory} $m^\pi\in\mathcal{M}^\pi$ , or recurrent state, and the current \newterm{observation} $o\in\mathcal{O}$.
Subscripts denote a player-specific component $s_i$, negative subscripts denote all but the player $s_{-i}$, and boldface denote the joint of all players $\bm{s}$.
The \newterm{transition dynamics} $p:\mathcal{H}\times\mathbfcal{A}\to\Delta(\mathcal{H})\times\Delta(\mathbfcal{R})$ define the game state update and reward signal.
The agent experiences \newterm{transitions}, or \newterm{experiences}, $(s^t,~a^t,r^{t+1},s^{t+1})$ of the game; where, sequences of transitions are called \newterm{trajectories}~$\tau$ and trajectories ending in a terminal game state are \newterm{episodes}.

At the start of an episode, all players sample their current \newterm{policy} $\pi$ from their \newterm{strategy} $\sigma:\Pi\to[0,1]$, where $\Pi$ is the \newterm{policy space} and $\Sigma$ is the corresponding \newterm{strategy space}. 
A \newterm{utility function} $U:\bm\Pi\to\mathbb{R}^n$ defines the payoffs/returns (i.e., cumulative reward) for each of $n$ players.
The tuple $\Gamma\equiv(\bm\Pi, U, n)$ defines a \newterm{normal-form game} (NFG) based on these elements.
We represent empirical games in normal form.
An \newterm{empirical normal-form game} (ENFG) $\hat\Gamma\equiv(\bm{\hat\Pi}, \hat{U}, n)$ models a game with a \newterm{restricted strategy set} $\bm{\hat\Pi}$ and an estimated payoff function $\hat{U}$.
An empirical game is typically built by alternating between game reasoning and strategy exploration.
During the game reasoning phase, the empirical game is solved based on a solution concept predefined by the modeler.
The strategy exploration step uses this solution to generate new policies to add to the empirical game.
One common heuristic is to generate new policies that best-respond to the current solution~\cite{mcmahan03do,schvartzman09exploration}.
As exact best-responses typically cannot be computed, RL or DRL are employed to derive approximate best-responses~\cite{lanctot17psro}.

An \newterm{agent world model} $w$ represents dynamics in terms of information available to the agent.
Specifically, $w$ maps information states and actions to observations and rewards,
$w:\mathbfcal{O}\times\mathbfcal{A}\times\mathcal{M}^w\to\mathbfcal{O}\times\mathbfcal{R}$, where $m^w\in\mathcal{M}^w$ is the \newterm{world model's memory}, or recurrent state.
For simplicity, in this work, we assume the agent learns and uses a deterministic world model, irrespective of stochasticity that may be present in the true game.
Specific implementation details for this work are provided in Appendix~\ref{app:world-model}.

Until now, we have implicitly assumed the need for distinct models. 
However, if a single model could serve both functions, co-learning two separate models would not be needed. 
Empirical games, in general, cannot replace a world model as they entirely abstract away any concept of game dynamics. 
Conversely, world models have the potential to substitute for the payoff estimations in empirical games by estimating payoffs as rollouts with the world model.
We explore this possibility in an auxiliary experiment included in Appendix~\ref{app:results:world-model-enfg}, but our findings indicate that this substitution is impractical. 
Due to compounding of model-prediction errors, the payoff estimates and entailed game solutions were quite inaccurate.

Having defined the models and established the need for their separate instantiations, we can proceed to evaluate the claims of beneficial co-learning. 
Our first experiment shows that the strategic diversity embodied in an empirical game yields diverse game dynamics, resulting in the training of a more performant world model. 
The second set of experiments demonstrates that a world model can help reduce the computational cost of policy construction in an empirical game.

\subsection{Strategic Diversity}
\label{sec:diversity}
A world model is trained to predict successor observations and rewards, from the current observations and actions, using a supervised learning signal.
Ideally, the training data would cover all possible transitions.
This is not feasible, so instead draws are conventionally taken from a dataset generated from play of a \newterm{behavioral strategy}.
Performance of the world model is then measured against a \newterm{target strategy}.
Differences between the behavioral and target strategies present challenges in learning an effective world model.

We call the probability of drawing a state-action pair under some strategy its \newterm{reach probability}.
From this, we define a strategy's \newterm{strategic diversity} as the distribution induced from reach probabilities. across the full state-action space.
These terms allow us to observe two challenges for learning world models.
First, the diversity of the behavioral strategy ought to \emph{cover} the target strategy's diversity.
Otherwise, transitions will be absent from the training data.
It is possible to supplement coverage of the absent transitions if they can be generalized from covered data; however, this cannot be generally guaranteed. 
Second, the \emph{closer} the diversities are, the more accurate the learning objective will be.
An extended formal argument of these challenges is provided in Appendix~\ref{app:diversity}.

If the target strategy were known, we could readily construct the ideal training data for the world model.
However the target is generally not known at the outset; indeed determining this target is the ultimate purpose of empirical game reasoning.
The evolving empirical game essentially reflects a search for the target.
Serendipitously, construction of this empirical game entails generation of data that captures elements of likely targets.
This data can be reused for world model training without incurring any additional data collection cost.

\paragraph{Game.}
We evaluate the claims of independent co-learning benefits within the context of a \textit{commons game} called ``Harvest''.
In Harvest, players move around an orchard picking apples. 
The challenging commons element is that apple regrowth rate is proportional to nearby apples, so that socially optimum behavior would entail managed harvesting. 
Self-interested agents capture only part of the benefit of optimal growth, thus non-cooperative equilibria tend to exhibit collective over-harvesting. 
The game has established roots in human-behavioral studies~\cite{janssen10commons} and in agent-based modeling of emergent behavior~\cite{perolat17commons, leibo17, leibo2021meltingpot}.
For our initial experiments, we use a symmetric two-player version of the game, where in-game entities are represented categorically~\cite{hcai19gathering}.
% This categorical representation facilitates faster experimentation and simplifies the interpretation of results. 
Each player has a $10\times10$ viewbox within their field of vision. 
The possible actions include moving in the four cardinal directions, rotating either way, tagging, or remaining idle.
A successful tag temporarily removes the other player from the game, but can only be done to other nearby players.
Players receive a reward of \num{1} for each apple picked.
More detailed information and visualizations are available in Appendix~\ref{app:games:harvest-categorical}.

\paragraph{Experiment.}
To test the effects of strategic diversity, we train a suite of world models that differ in the diversity of their training data.
The datasets are constructed from the play of three policies: a random baseline policy, and two PSRO-generated policies.
The PSRO policies were arbitrarily sampled from an approximate solution produced by a run of PSRO\@.
We sampled an additional policy from PSRO for evaluating the generalization capacity of the world models.
These policies are then subsampled and used to train seven world models.
The world models are referred to by icons \profile{}\! that depict the symmetric strategy profiles used to train them in the normal-form.
Strategy profiles included in the training data of the world models are shaded black. 
For instance, the first (random) policy \wmA\!\!, or the first and third policies \wmAC\!\!.
Each world model's dataset contains \num{1} million total transitions, collected uniformly from each distinct strategy profile (symmetric profiles are not re-sampled).
The world models are then evaluated on accuracy and recall for their predictions of both observation and reward for both players.
The world models are optimized with a weighted-average cross-entropy objective.
Additional details are in Appendix~\ref{app:world-model}.

\begin{figure}[!ht]
    \centering
    \includegraphics{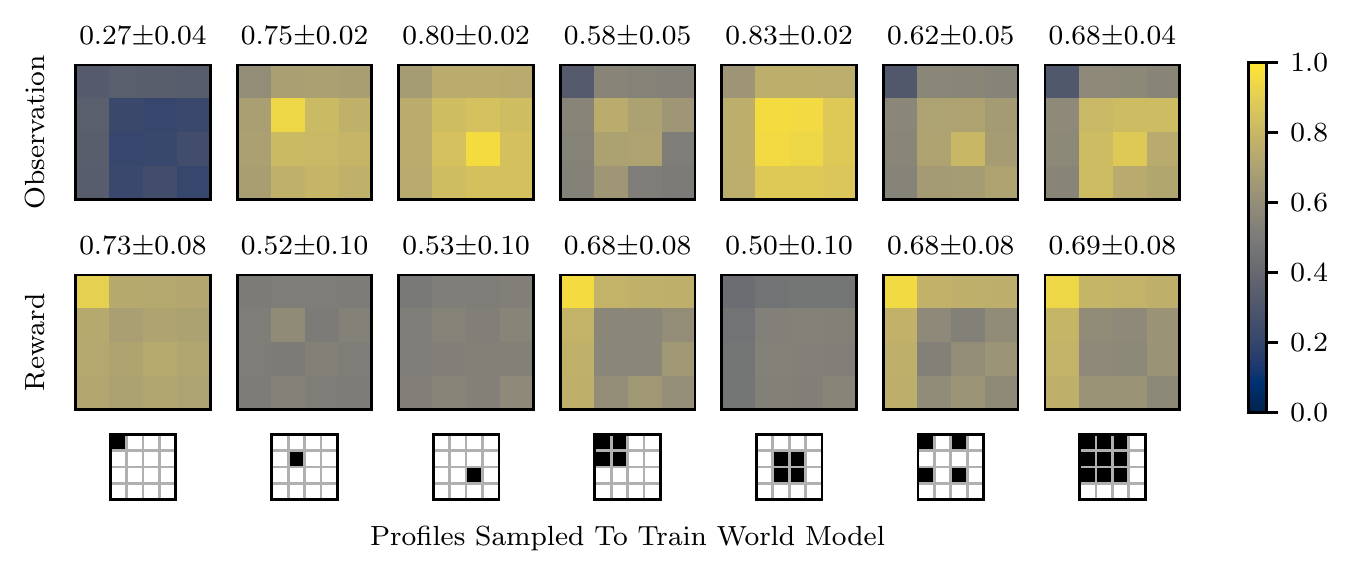}
    \caption{%
        World model accuracy across strategy profiles. 
        Each heatmap portrays a world model's accuracy over 16 strategy profiles.
        The meta x-axis corresponds to the profiles used to train the world model (as black cells).
        Above each heatmap is the model's average accuracy.
    }
    \label{fig:world-model-comparison-accuracy}
\end{figure}

\paragraph{Results.}
Figure~\ref{fig:world-model-comparison-accuracy} presents each world model's per-profile accuracy, as well as its average over all profiles.
Inclusion of the random policy corresponds to decreases in observation prediction accuracy: 
\mbox{\wmB\num{0.75\pm0.02} $\to$ \wmAB\num{0.58\pm0.05}}, 
\mbox{\wmC\num{0.80\pm0.02} $\to$ \wmAC\num{0.62\pm0.05}}, and 
\mbox{\wmBC\num{0.83\pm0.02} $\to$ \wmABC\num{0.68\pm0.04}}.
Figure~\ref{fig:gathering_recall}~(Appendix~\ref{app:results:diversity}) contains the world model's per-profile recall.
Inclusion of the random policy corresponds to increases in reward~\num{1} recall: 
\mbox{\wmB\num{0.25\pm0.07} $\to$ \wmAB\num{0.37\pm0.11}},
\mbox{\wmC\num{0.25\pm0.07} $\to$ \wmAC\num{0.36\pm0.11}}, and
\mbox{\wmBC\num{0.26\pm0.07} $\to$ \wmABC\num{0.37\pm0.11}}.

\paragraph{Discussion.}
The PSRO policies offer the most strategically salient view of the game's dynamics. 
Consequently, the world model \wmBC trained with these policies yields the highest observation accuracy.
However, this world model performs poorly on reward accuracy, scoring only \num{0.50\pm0.10}. 
In comparison, the model trained on the random policy \wmA scores \num{0.73\pm0.08}.
This seemingly counterintuitive result can be attributed to a significant class imbalance in rewards. 
\wmA predicts only the most common class, no reward, which gives the illusion of higher performance. 
In contrast, the remaining world models attempt to predict rewarding states, which reduces their overall accuracy.
Therefore, we should compare the world models based on their ability to recall rewards. 
When we examine \wmBC again, we find that it also struggles to recall rewards, scoring only \num{0.26\pm0.07}. 
However, when the random policy is included in the training data (\!\wmABC\!\!), the recall improves to \num{0.37\pm0.11}.
This improvement is also due to the same class imbalance. 
The PSRO policies are highly competitive, tending to over-harvest.
This limits the proportion of rewarding experiences.
Including the random policy enhances the diversity of rewards in this instance, as its coplayer can demonstrate successful harvesting.
Given the importance of accurately predicting both observations and rewards for effective planning, \wmABC appears to be the most promising option. 
However, the strong performance of \wmBC suggests future work on algorithms that can benefit solely from observation predictions. 
Overall, these results support the claim that strategic diversity enhances the training of world models.

\subsection{Response Calculations}
\label{sec:response-calculation}
Empirical games are built by iteratively calculating and incorporating responses to the current solution.
However, direct computation of these responses is often infeasible, so RL or DRL is used to approximate the response.
This process of approximating a single response policy using RL is computationally intensive, posing a significant constraint in empirical game modeling when executed repeatedly.
World models present an opportunity to address this issue.
A world model can serve as a medium for transferring previously learned knowledge about the game's dynamics.
Therefore, the dynamics need not be relearned, reducing the computational cost associated with response calculation. 

Exercising a world model for transfer is achieved through a process called \newterm{planning}.
Planning is any procedure that takes a world model and produces or improves a policy.
In the context of games, planning can optionally take into account the existence of coplayers.
This consideration can reduce experiential variance caused by unobserved confounders (i.e., the coplayers). 
However, coplayer modeling errors may introduce further errors in the planning procedure~\cite{he16opponentmodeling}.

Planning alongside empirical-game construction allows us to side-step this issue as we have direct access to the policies of all players during training. 
This allows us to circumvent the challenge of building accurate agent models. 
Instead, the policies of coplayers can be directly queried and used alongside a world model, leading to more accurate planning. 
In this section, we empirically demonstrate the effectiveness of two methods that decrease the cost of response calculation by integrating planning with a world model and other agent policies.

\subsubsection{Background Planning}
\label{sec:background-planning}
The first type of planning that is investigated is \newterm{background planning}, popularized by the Dyna architecture~\cite{sutton90dyna}.
In background planning, agents interact with the world model to produce \newterm{planned experiences}\footnote{%
Other names include ``imaginary'', ``simulated'', or ``hallucinated'' experiences.
}.
The planned experiences are then used by a model-free reinforcement learning algorithm as if they were \newterm{real experiences} (experiences generated from the real game).
Background planning enables learners to generate experiences of states they are not currently in.

\paragraph{Experiment.}
To assess whether planned experiences are effective for training a policy in the actual game, we compute two response policies. 
The first response policy, serving as our baseline, learns exclusively from real experiences. 
The second response policy, referred to as the planner, is trained using a two-step procedure. 
Initially, the planner is exclusively trained on planned experiences. 
After \num{10000} updates, it then transitions to learning solely from real experiences. 
Policies are trained using IMPALA~\cite{espeholt18impala}, with further details available in Appendix~\ref{app:policy}.
The planner employs the \wmABC world model from Section~\ref{sec:diversity}, and the opponent plays the previously held-out policy. 
In this and subsequent experiments, the cost of methods is measured by the number of experiences they require with the actual game. 
This is because, experience collection is often the bottleneck when applying RL-based methods~\cite{obando2021revisitingrainbow, hester2012texplore}. 
Throughout the remainder of this work, each experience represents a trajectory of \num{20} transitions, facilitating the training of recurrent policies.

\begin{figure}[!ht]
    \centering
    \includegraphics{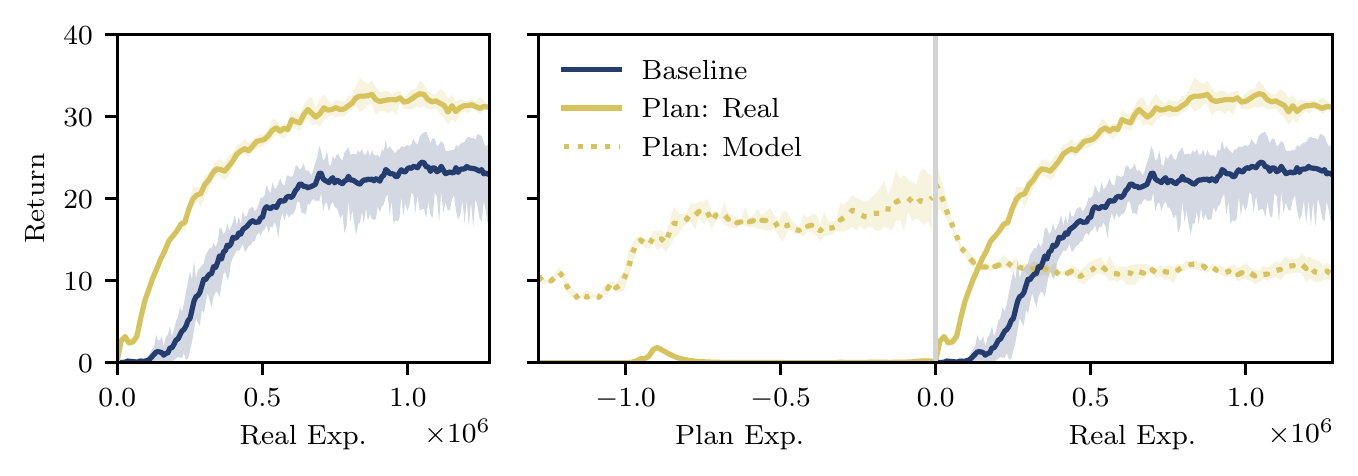}
    \caption{%
        Effects of background planning on response learning.
        Left: Return curves measured by the number of real experiences used.
        Right: Return curves measured by usage of both real and planned experiences.
        The planner's return is measured against the real game and the world model.
        (\num{5} seeds, with \SI{95}{\percent} bootstrapped CI).
    }
    \label{fig:background-planning}
\end{figure}

% Baseline: \num{23.00\pm4.01}
% WM: Good: \num{31.17\pm0.25}
% WM: Bad:  \num{26.05\pm1.32}
\paragraph{Results.}
Figure~\ref{fig:background-planning} presents the results of the background planning experiment. 
The methods are compared based on their final return, utilizing an equivalent amount of real experiences. 
The baseline yields a return of \num{23.00\pm4.01}, whereas the planner yields a return of \num{31.17\pm0.25}.

\paragraph{Discussion.}
In this experiment, the planner converges to a stronger policy, and makes earlier gains in performance than the baseline.
Despite this, there is a significant gap in the planner's learning curves, which are reported with respect to both the world model and real game.
This gap arises due to accumulated model-prediction errors, causing the trajectories to deviate from the true state space. 
Nevertheless, the planner effectively learns to interact with the world model during planning, and this behavior shows positive transfer into the real game, as evidenced by the planner's rapid learning.
The exact magnitude of benefit will vary across coplayers' policies, games, and world models.
In Figure~\ref{fig:background-planning-bad}~(Appendix~\ref{app:results:background-planning}), we repeat the same experiment with the poorly performing \wmA world model, and observe a marginal benefit (\num{26.05\pm1.32}).
The key take-away is that background planning tends to lead towards learning benefits, and not generally hamper learning.

\subsubsection{Decision-Time Planning}
\label{sec:decision-time-planning}
The second main way that a world model is used is to inform action selection at \newterm{decision time [planning] (DT)}.
In this case, the agent evaluates the quality of actions by comparing the value of the model's predicted successor state for all candidate actions.
Action evaluation can also occur recursively, allowing the agent to consider successor states further into the future. 
Overall, this process should enable the learner to select better actions earlier in training, thereby reducing the amount of experiences needed to compute a response.
A potential flaw with decision-time planning is that the agent's learned value function may not be well-defined on model-predicted successor states~\cite{talvitie14}.
To remedy this issue, the value function should also be trained on model-predicted states.

\paragraph{Experiment.}
To evaluate the impact the decision-time planning, we perform an experiment similar to the background planning experiment (Section~\ref{sec:background-planning}).
However, in this experiment, we evaluate the quality of four types of decision-time planners that perform one-step three-action search.
The planners differ in the their ablations of background planning types: 
(1)~\newterm{warm-start background planning (BG:~W)} learning from planned experiences before any real experiences, and 
(2)~\newterm{concurrent background planning (BG:~C)} where after BG:~W, learning proceeds simultaneously on both planned and real experiences.
The intuition behind BG:~C is that the agent can complement its learning process by incorporating planned experiences that align with its current behavior, offsetting the reliance on costly real experiences.
Extended experimental details are provided in Appendix~\ref{app:methods}.

\begin{figure}[!ht]
    \centering
    \includegraphics{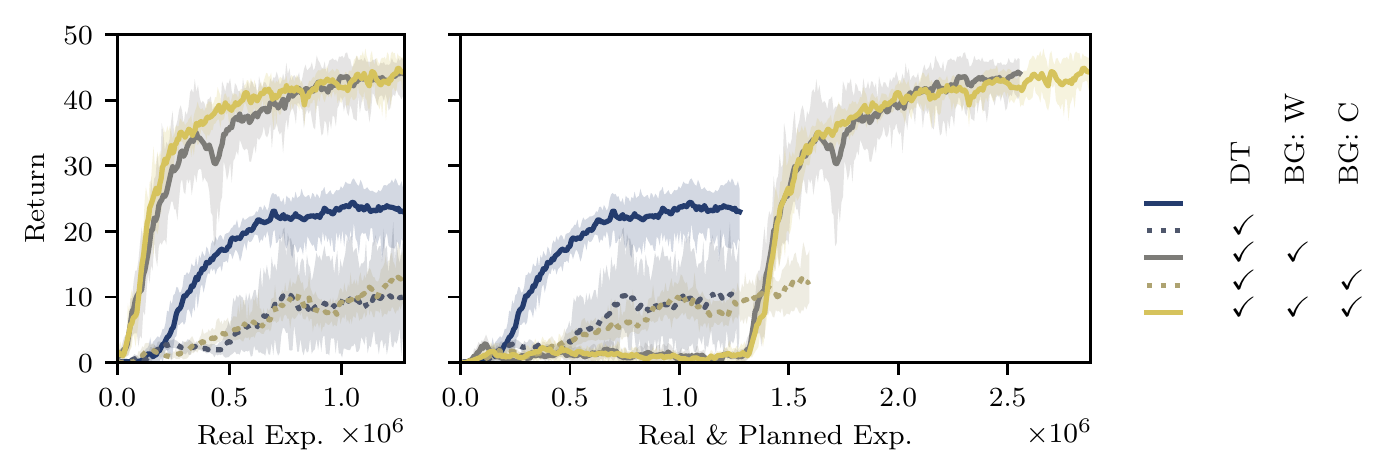}
    \caption{%
        Effects of decision-time planning on response learning.
        Four planners using decision-time planning (DT) are shown in combinations with warm-start background planning (BG:~W) and concurrent background planning (BG:~C).
        (\num{5} seeds, with \SI{95}{\percent} bootstrapped CI).
    }
    \label{fig:decision-time-planning}
\end{figure}

% GOOD ============================
% baseline      \num{23.00\pm4.01}
% search        \num{9.98\pm7.60}
% search-warm   \num{44.11\pm2.81}
% search-bg     \num{12.42\pm3.97}
% search-both   \num{44.31\pm2.56}
% BAD  ============================
% baseline      \num{23.00\pm4.01}
% search        \num{6.29\pm5.12}
% search-warm   \num{20.98\pm9.76}
% search-bg     \num{3.64\pm0.26}
% search-both   \num{33.07\pm7.67}
\paragraph{Results.}
The results for this experiment are shown in Figure~\ref{fig:decision-time-planning}.
The baseline policy receives a final return of \num{23.00\pm4.01}.
The planners that do not include BG:~W, perform worse, with final returns of \num{9.98\pm7.60} (DT) and \num{12.42\pm3.97} (DT \& BG:~C).
The planners that perform BG:~W outperform the baseline, with final returns of \num{44.11\pm2.81} (DT \& BG:~W) and \num{44.31\pm2.56} (DT, BG:~W, \& BG:~C).

\paragraph{Discussion.}
Our results suggest that the addition of BG:~W provides sizable benefits: \num{9.98\pm7.60} (DT) $\to$ \num{44.11\pm2.81} (DT \& BG:W) and \num{12.42\pm3.97} (DT \& BG:~C) $\to$ \num{44.31\pm2.56} (DT, BG:~W, \& BG:~C).
We postulate that this is because it informs the policy's value function on model-predictive states early into training.
This allows that the learner is able to more effectively search earlier into training.
BG:~C appears to offer minor stability and variance improvements throughout the training procedure; however, it does not have a measurable difference in final performance.
This result suggests using planning methods in combination to reap their respective advantages.

However, we caution against focusing on the magnitude of improvement found within this experiment.
As the margin of benefit depends on many factors including the world model accuracy, the opponent policy, and the game.
To exemplify, similar to the background planning section, we repeat the same experiment with the poorly performing \wmA world model.
The results of this ancillary experiment are in Figure~\ref{fig:decision-time-planning-bad} (Appendix~\ref{app:results:decision-time-planning}).
The trend of BG:~W providing benefits was reinforced: \num{6.29\pm5.12} (DT) $\to$ \num{20.98\pm9.76} (DT \& BG:~W) and \num{3.64\pm0.26} (DT \& BG:~C) $\to$ \num{33.07\pm7.67} (DT, BG:~W, \& BG:~C).
However, the addition of BG:~C now measurably improved performance \num{20.98\pm9.76} (DT \& BG:~W) $\to$ \num{33.07\pm7.67} (DT, BG:~W, \& BG:~C).
The main outcome of these experiments is the observation that multi-faceted planning is unlikely to harm a response calculation, and has a potentially large benefit when applied effectively.
These results support the claim that world models offer the potential to improve response calculation through decision-time planning.

\section{Dyna-PSRO}
\label{sec:dyna-psro}
In this section we introduce Dyna-PSRO, \emph{Dyna}-Policy-Space Response Oracles, an approximate game-solving algorithm that builds on the PSRO~\cite{lanctot17psro} framework. 
Dyna-PSRO employs co-learning to combine the  benefits of world models and empirical games.

Dyna-PSRO is defined by two significant alterations to the original PSRO algorithm.
First, it trains a world model in parallel with all the typical PSRO routines (i.e., game reasoning and response calculation).
We collect training data for the world model from both the episodes used to estimate the empirical game's payoffs, and the episodes that are generated during response learning and evaluation.
This approach ensures that the world model is informed by a diversity of data from a salient set of strategy profiles.
By reusing data from empirical game development, training the world model incurs no additional cost for data collection.

The second modification introduced by Dyna-PSRO pertains to the way response policies are learned.
Dyna-PSRO adopts a Dyna-based reinforcement learner~\cite{sutton90dyna, sutton91dyna, sutton12lineardyna} that integrates simultaneous planning, learning, and acting.
Consequently, the learner concurrently processes experiences generated from decision-time planning, background planning, and direct game interaction.
These experiences, regardless of their origin, are then learned from using the IMPALA~\cite{espeholt18impala} update rule.
For all accounts of planning, the learner uses the single world model that is trained within Dyna-PSRO.
This allows game knowledge accrued from previous response calculations to be transferred and used to reduce the cost of the current and future response calculations.
Pseudocode and additional details for both PSRO and Dyna-PSRO are provided in Appendix~\ref{app:dyna-psro}.

\paragraph{Games.}
Dyna-PSRO is evaluated on three games.
The first is the harvest commons game used in the experiments described above, denoted ``Harvest: Categorical''.
The other two games come from the MeltingPot~\cite{leibo2021meltingpot} evaluation suite and feature rich image-based observations.
``Harvest: RGB'' is their version of the same commons harvest game (details in Appendix~\ref{app:games:harvest-rgb}).
``Running With Scissors'' is a temporally extended version of rock-paper-scissors (details in Appendix~\ref{app:games:running-with-scissors}).
World model training and implementation details for each game are in Appendix~\ref{app:world-model}, likewise, policies in Appendix~\ref{app:policy}.

\paragraph{Experiment.}
Dyna-PSRO's performance is measured by the quality of the solution it produces when compared against the world-model-free baseline PSRO\@.
The two methods are evaluated on SumRegret (sometimes called \emph{Nash convergence}), which measures the regret across all players $\text{SumRegret}(\bm{\sigma}, \overline{\bm{\Pi}})=\sum_{i\in n}\max_{\pi_i\in\overline\Pi_i}\hat{U}_i(\pi_i, \sigma_{-i}) - \hat{U}_i(\sigma_i, \sigma_{-i})$, where $\bm{\sigma}$ is the method's solution and $\overline{\bm{\Pi}}\subseteq\bm{\Pi}$ denotes the deviation set. 
We define deviation sets based on policies generated across methods (i.e., regret is with respect to the \emph{combined game}): $\overline{\bm{\Pi}}\equiv\bigcup_\textup{method}\hat{\bm{\Pi}}^\textup{method}$, for all methods for a particular seed (detailed in Appendix~\ref{app:regret})~\cite{balduzzi2018reevaluating}.
We measure SumRegret for intermediate solutions, and report it as a function of the cumulative number of real experiences employed in the respective methods.

\begin{figure}[!ht]
    \centering
    \includegraphics{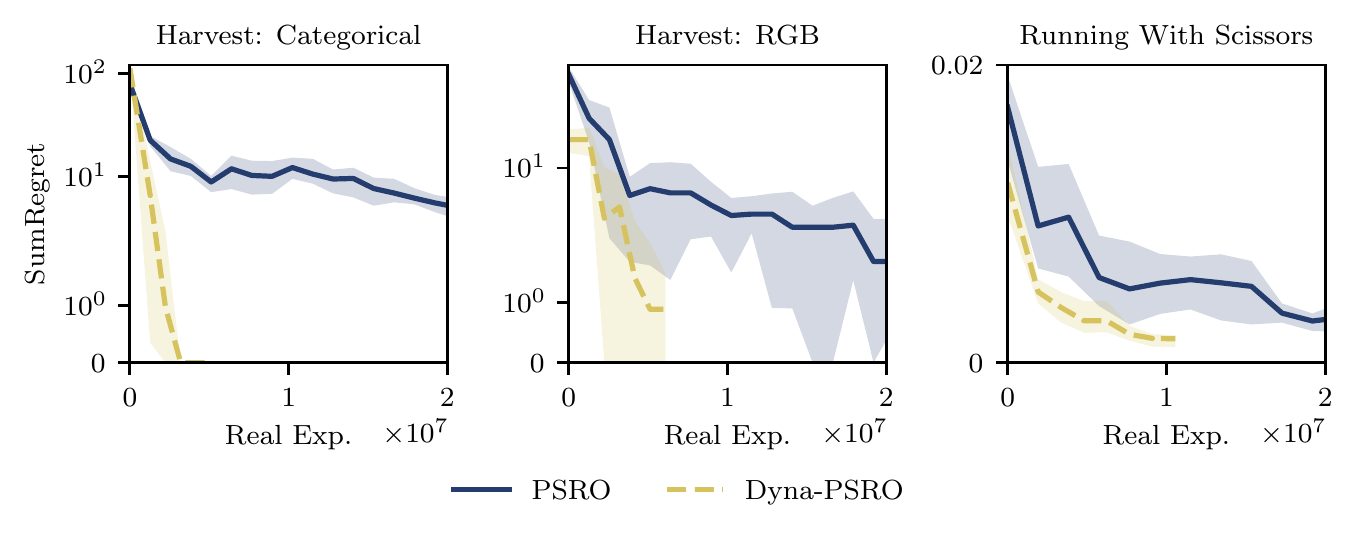}
    \caption{%
        PSRO compared against Dyna-PSRO. 
        (\num{5} seeds, with \SI{95}{\percent} bootstrapped CI).        
    }
    \label{fig:dyna-psro}
\end{figure}

% Harvest: Categorical
%     Best solution in \num{3.2e6}
% Harvest: RGB
%     \num{0.89\pm0.74} within \num{5.12e6}
%     \num{6.42\pm4.73} PSRO at same time
%     \num{2.50\pm2.24} at end \num{2.05e7}
% RWS 
%     Dyna
%     $2\mathrm{e}{-3}\pm5\mathrm{e}{-4}$ at \num{1.06e7}
%     PSRO
%     $6.68\mathrm{e}{-3}\pm2.51\mathrm{e}{-3}$ at \num{9.6e6}
%     $3.50\mathrm{e}{-3}\pm7.36\mathrm{e}{-4}$$ 
\paragraph{Results.}
Figure~\ref{fig:dyna-psro} presents the results for this experiment.
For Harvest:~Categorical, Dyna-PSRO found a no regret solution within the combined-game in \num{3.2e6} experiences.
Whereas, PSRO achieves a solution of at best \num{5.45\pm1.62} within \num{2e7} experiences.
In Harvest:~RGB, Dyna-PSRO reaches a solution with \num{0.89\pm0.74} regret at \num{5.12e6} experiences.
At the same time, PSRO had found a solution with \num{6.42\pm4.73} regret, and at the end of its run had \num{2.50\pm2.24} regret.
In the final game, RWS, Dyna-PSRO has $2\mathrm{e}{-3}\pm5\mathrm{e}{-4}$ regret at \num{1.06e7} experiences, and at a similar point (\num{9.6e6} experiences), PSRO has $6.68\mathrm{e}{-3}\pm2.51\mathrm{e}{-3}$.
At the end of the run, PSRO achieves a regret $3.50\mathrm{e}{-3}\pm7.36\mathrm{e}{-4}$.

\paragraph{Discussion.}
The results indicate that across all games, Dyna-PSRO consistently outperforms PSRO by achieving a superior solution. 
Furthermore, this improved performance is realized while consuming fewer real-game experiences. 
For instance, in the case of Harvest:~Categorical, the application of the world model for decision-time planning enables the computation of an effective policy after only a few iterations. 
On the other hand, we observe a trend of accruing marginal gains in other games, suggesting that the benefits are likely attributed to the transfer of knowledge about the game dynamics.
In Harvest:~Categorical and Running With Scissors, Dyna-PSRO also had lower variance than PSRO.

\section{Limitations}

Although our experiments demonstrate benefits for co-learning world models and empirical games, there are several areas for potential improvement. 
The world models used in this study necessitated observational data from all players for training, and assumed a simultaneous-action game. 
Future research could consider relaxing these assumptions to accommodate different interaction protocols, a larger number of players, and incomplete data perspectives.
Furthermore, our world models functioned directly on agent observations, which made them computationally costly to query.
If the generation of experiences is the major limiting factor, as assumed in this study, this approach is acceptable. 
Nevertheless, reducing computational demands through methods like latent world models presents a promising avenue for future research. 
Lastly, the evaluation of solution concepts could also be improved. 
While combined-game regret employs all available estimates in approximating regret, its inherent inaccuracies may lead to misinterpretations of relative performance.

\section{Conclusion}
This study showed the mutual benefit of co-learning a world model and empirical game.
First, we demonstrated that empirical games provide strategically diverse training data that could inform a more robust world model.
We then showed that world models can reduce the computational cost, measured in experiences, of response calculations through planning.
% A key operation in building an empirical game.
These two benefits were combined and realized in a new algorithm, Dyna-PSRO\@.
In our experiments, Dyna-PSRO computed lower-regret solutions than PSRO on several partially observable general-sum games.
Dyna-PSRO also required substantially fewer experiences than PSRO, a key algorithmic advantage for settings where collecting experiences is a cost-limiting factor.

%% file: appendix.tex
\section{Broader Impact}
\label{app:broader-impact}
\input{appendices/broader_impact}

\section{Compute}
\label{app:compute}
\input{appendices/compute}

\section{Methods Details}
\label{app:methods}
In this work, the both the policies and world models are implemented in JAX~\cite{jax2018github} with Haiku~\cite{haiku2020github}.
The software is architected using Launchpad~\cite{yang2021launchpad} with design patterns inspired by ACME~\cite{hoffman2020acme}.
All replay buffers are implemented using Reverb~\cite{cassirer2021reverb}.
Gambit~\cite{gambit} is used as a game solver via linear complementarity~\cite{eaves71lcp}.

\subsection{Policy Implementation \& Training}
\label{app:policy}
\input{appendices/methods/policy}

\subsection{World Model Implementation \& Training}
\label{app:world-model}
\input{appendices/methods/world_model}

\subsection{Strategic Diversity}
\label{app:diversity}
\input{appendices/methods/strategic_diversity}

\subsection{Dyna-PSRO}
\label{app:dyna-psro}
\input{appendices/methods/dyna_psro}

\subsection{Combined-Game Regret}
\label{app:regret}
\input{appendices/methods/regret}

\clearpage
\section{Games}
\label{app:games}

\subsection{Harvest: Categorical}
\label{app:games:harvest-categorical}
\input{appendices/games/harvest_categorical}

\subsection{Harvest: RGB}
\label{app:games:harvest-rgb}
\input{appendices/games/harvest_rgb}

\subsection{Running With Scissors}
\label{app:games:running-with-scissors}
\input{appendices/games/running_with_scissors}

\clearpage
\section{Additional Results}
\label{app:results}

\subsection{Strategic Diversity}
\label{app:results:diversity}
\input{appendices/results/diversity}

\clearpage
\subsection{Background Planning}
\label{app:results:background-planning}
\input{appendices/results/background_planning}

\subsection{Decision-Time Planning}
\label{app:results:decision-time-planning}
\input{appendices/results/decision_time_planning}

\subsection{World Models as Empirical Games}
\label{app:results:world-model-enfg}
\input{appendices/results/world_model_enfg}

% \clearpage
% \section{Symbols}
% \label{app:symbols}
% \input{appendices/symbols}

%% file: appendices/broader_impact.tex
There are no direct broader impacts from this work.
However, this work promotes the adoption of both empirical games and world models.
Potential negative impacts may arise due to errors introduced when compressing the true game into the confines of \emph{any} model, which could lead to negative consequences.
World model errors within Dyna-PSRO are transferred across response calculations potentially reinforcing biases about the world. 
If these biases are not rectified, they could negatively influence policies learned from these models.
The strategic diversity component of this work aims to mitigate these potential biases, though it represents only the initial step in addressing this concern. 
When considering empirical games, inaccuracies within them can lead to the suggestion of flawed solutions. 
The adoption of these inaccurate solutions could have negative repercussions for practitioners or other stakeholders involved in the game. 
Vigilance and thorough evaluation are required to prevent these potential issues.

%% file: appendices/compute.tex
GPUs are used for training world models, and policies within Dyna-PSRO.
Two types of GPUs were used throughout this work interchangeably: TITAN~X and GTX~1080~Ti.
All other computation was completed using CPUs.
Each response calculation had additional CPUs corresponding to the number of experience generation arenas described in Appendix~\ref{app:methods}.
Experiments were run on internal clusters.

%% file: appendices/methods/policy.tex
\begin{figure}[ht]
    \centering
    \includegraphics{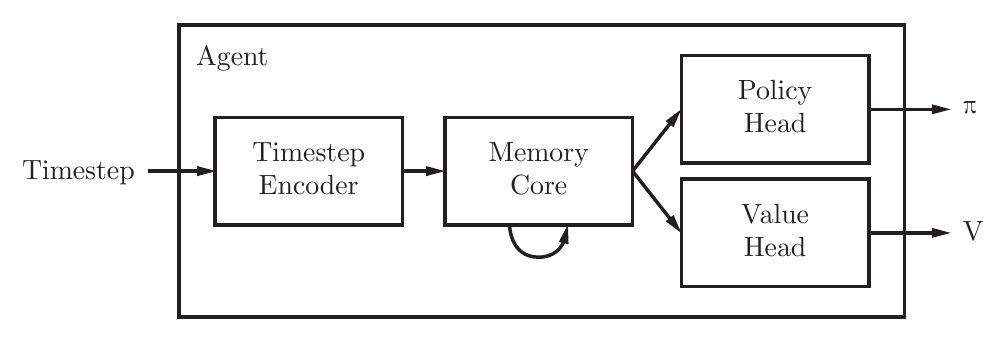}
    \caption{Agent Architecture.}
    \label{fig:agent-implementation}
\end{figure}

All policies follow the general architecture depicted in Figure~\ref{fig:agent-implementation}.
This consists of four modules:
\begin{itemize}
    \item\emph{Timestep Encoder}: 
    Processes all of the current observation's information into a single embedding vector. 
    The timestep includes the new observation and the policy's previous action.
    
    \item\emph{Memory Core}: The component of the agent that maintains and update's the agent's memory.
    
    \item\emph{Policy Head}: Computes the agent's policy.
    
    \item\emph{Value Head}: Computes the agent's state value function.
\end{itemize}
All of the components are simultaneously trained and their joint parameters are $\theta^\pi\in\Theta^\pi$.
The policies are trained using the IMPALA algorithm~\cite{espeholt18impala}.
For the IMPALA loss, the coefficients for each component loss are:
\begin{equation*}
    \mathcal{L}_\textup{IMPALA} = \lambda_\pi \cdot \mathcal{L}_\pi + \lambda_\textup{V}\cdot \mathcal{L}_\textup{V} + \lambda_\textup{entropy}\cdot \mathcal{L}_\textup{entropy},
\end{equation*}
with a discount factor of \num{0.99}.
% Hyperparameter tuning was only performed on the baseline response method to ensure stable convergence.
% All methods use approximately the same policy modules that are implemented as follows:
The training details for each specific response calculation are itemized below.

\paragraph{Baseline Parameters}
The learning rate begins is linearly decayed over \num{10000} updates.
Each update is computed from a mini-batch of \num{128} examples that are generated from \num{8} arenas\footnote{%
The term \emph{arena} is used to refer to an experience generation process. 
This is more commonly referred to as an ``actor''; however, this terminology may be confounding with language in RL, Dyna, or multiagent learning. 
}
Policy parameters are synchronized at the beginning of each episode.
Each example in the mini-batch is a sequence of \num{20} transitions.
Moreover, sequences are stored in a replay buffer with a period of \num{19}, to ensure that the action played at the end of a sequence is trained.
Sequences are stored in a replay buffer with a max capacity of \num{1000000}, and are evicted once sampled.
Additional hyperparameters are specified in Table~\ref{tab:hyperparameters}.

\begin{table}[htbp]
    \centering
    \caption{%
        Baseline policy hyperparameters per game.
    }      
    \label{tab:hyperparameters}
    \begin{tabular}{lrrr}
        \toprule
         Hyperparameter &Harvest:~Categorical &Harvest:~RGB &Running with Scissors \\
         \midrule
         Optimizer &Adam~\cite{kimga14adam} &RMSProp~\cite{hinton18rmsprop} &RMSProp~\cite{hinton18rmsprop} \\
         $\lambda_\pi$ &\num{1.0} &\num{1.0} &\num{1.0} \\
         $\lambda_\textup{V}$ &\num{0.2} &\num{0.5} &\num{0.2} \\
         $\lambda_\textup{entropy}$ &\num{0.04} &\num{0.01} & \num{0.003}\\
         % Timestep Encoder &$\textup{MLP}(512, 256)$ && \\
         % Memory Core &$\textup{LSTM}(256)$ &$\textup{LSTM}(128)$ &$\textup{LSTM}(128)$  \\
         % Policy Head &$\textup{Linear}(8)$ &$\textup{Linear}(8)$ &$\textup{Linear}(8)$ \\
         % Value Head &$\textup{Linear}(1)$ &$\textup{Linear}(1)$ &$\textup{Linear}(1)$ \\
         Learning Rate Start &\num{6e-6} &\num{6e-4} &\num{1e-4} \\
         Learning Rate Stop &\num{6e-9} &\num{6e-9} &\num{1e-4} \\
         Max Grad Norm &\num{10.0} &\num{1.0} &\num{0.1} \\
         Batch Size &\num{128} &\num{128} &\num{128} \\     
         \bottomrule
    \end{tabular}
\end{table}

Harvest:~Categorical module implementations:
\begin{itemize}
    \item\emph{Timestep Encoder}:
    The encoder processes two timestep components: the current observation and the previous action the policy took. 
    First the observation is passed through a two-layer fully connected neural network with hidden sizes of $[256, 256]$.
    The representation of the observation is then concatenated with the previous action (represented as a one-hot vector), and passed together through a second neural network with sizes $[256, 256]$. 
    All of the layers have ReLU~\cite{fukushima75cognitron} activations including the final layers of both networks. The final representation is the output of the timestep encoder.
    
    \item\emph{Memory Core}:
    A single-layer LSTM~\cite{hochreiter97lstm} with \SI{256}{units}.
    
    \item\emph{Policy Head}: 
    A single linear layer of size \num{8}.
    
    \item\emph{Value Head}: 
    A single linear layer of size \num{1}.
\end{itemize}

Harvest:~RGB and Running with Scissors module implementations:
\begin{itemize}
    \item\emph{Timestep Encoder}:
    The encoder processes two timestep components: the current observation and the previous action the policy took. 
    The observation is first process by a two-layer convolutional neural network with ReLU activations~\cite{fukushima75cognitron}.
    The first layer has 16 channels, a kernel with shape $[8, 8]$, and a stride of $[8, 8]$.
    The second layer has 32 channels, a kernel shape of $[4, 4]$, and a stride of $[1, 1]$.
    The output of this layer is then flattened and concatenated with a one-hot encoding of the policy's previous action.
    The resulting embedding is then passed through a two-layer fully connected neural network with hidden sizes of $[128, 128]$, and ReLU activations.
    
    \item\emph{Memory Core}:
    A single-layer LSTM~\cite{hochreiter97lstm} with \SI{128}{units}.
    
    \item\emph{Policy Head}: 
    A single linear layer of size \num{8}.
    
    \item\emph{Value Head}: 
    A single linear layer of size \num{1}.
\end{itemize}

\paragraph{Planning Parameters}
The planners have the same hyperparameters as the baseline method, but with the addition of planning-specific settings.
For all planners, an additional \num{4} arenas are used to generate planned experiences (for background planning).
The additional settings for each version of planning are as follows:
\begin{itemize}
    \item\emph{Warm-Start Background Planning}:
    An additional \num{10000} updates are performed on exclusively planned experiences before play in the real game occurs.
    
    \item\emph{Concurrent Background Planning}:
    Each mini-batch sampled after warm-starting contains \SI{25}{\percent} planned experiences, and \SI{75}{\percent} real experiences.
    
    \item\emph{Decision-Time Planning}:
    In the training arenas (those that have the real game, and are not used for evaluation), the agent selects actions with a beam-search of width \num{3} and depth \num{1}.
\end{itemize}
Background planning also requires defining a \emph{search control} procedure~\cite{sutton90dyna, sutton91dyna, sutton18book}.
Search control defines how the agent prioritizes selecting starting states and actions for background planning.
This work considers the simplest search-control method: maintain a buffer of the initial states and uniformly sample.

%% file: appendices/methods/world_model.tex
\subsubsection{Action-Conditioned Scheduled Sampling}
\label{app:scheduled-sampling}

\begin{wrapfigure}[11]{R}{0.46\textwidth}
\begin{algorithm}[H]
\SetAlgoNoLine
\DontPrintSemicolon
\caption{Action-Conditioned\\ Scheduled Sampling}
\label{alg:action-biased-scheduled-sampling}
$m\gets$ Initial recurrent state\;
\For{$t \in T$}{
    $\bm{o} \gets \bm{o}^t$ if $\textup{Unif}[0, 1]<\epsilon(t)$ else $\bm{\hat{o}}^t$\;
    $\bm{\hat{o}}^{t+1}, \bm{\hat{r}}^{t+1}, m\gets w(\bm{o}, \bm{a}^t, m)$\;
}
\KwOut{Predicted trajectory $(\bm{\hat{o}}^{0:T}, \bm{\hat{r}}^{0:T})$}
\end{algorithm}
\end{wrapfigure}

As noted by Talvitie~\cite{talvitie14}, rolling out trajectories with an imperfect model tends to result in compounding errors in prediction.
Their work suggests training a Markovian world model with previous predictions (referred to as ``hallucinated replay''), to train the model to correct errors.
For stateful world models, as studied in this work, it has been demonstrated that curricula of $n$-step future predictions can train a fruitful world model~\cite{michalski14nstep, oh15prediction, chiappa17}.
All of the preceding work was studying single-agent systems; therefore, they could assume a much more stable data distribution for training.
As a result, these fixed curricula style approaches may prove fatal as the data distribution may change dramatically throughout training based on the coplayers' strategies. 

Instead, this work adapts the scheduled sampling~\cite{bengio15scheduledsampling} algorithm as a stochastic curricula, which will allow both short- and long-term predictions throughout the course of training.
Scheduled sampling is an algorithm for training auto-regressive sequence prediction models where at each predictive step during training the model input is sampled from either the previous prediction or the ground truth.
Adapting this algorithm for world model rollouts requires biasing each predictive step with the true actions while sampling between the predicted successor observation and the true successor observation.
Therefore, the predictions will always be biased on true actions, but must learn to handle model-predicted observation. 
The sampling follows a schedule $\epsilon: \mathbb{Z}\to[0, 1]$ that determines the probability of sampling the true observation over the previous prediction.
When $\epsilon$ is \num{1.0}, the algorithm behaves akin to teacher forcing~\cite{williams89teacherforcing} (with the same action-conditional modification); whereas, as it approaches \num{0.0} it becomes fully auto-regressive.

\subsubsection{Implementation}

\begin{figure}[ht]
    \centering
    \includegraphics{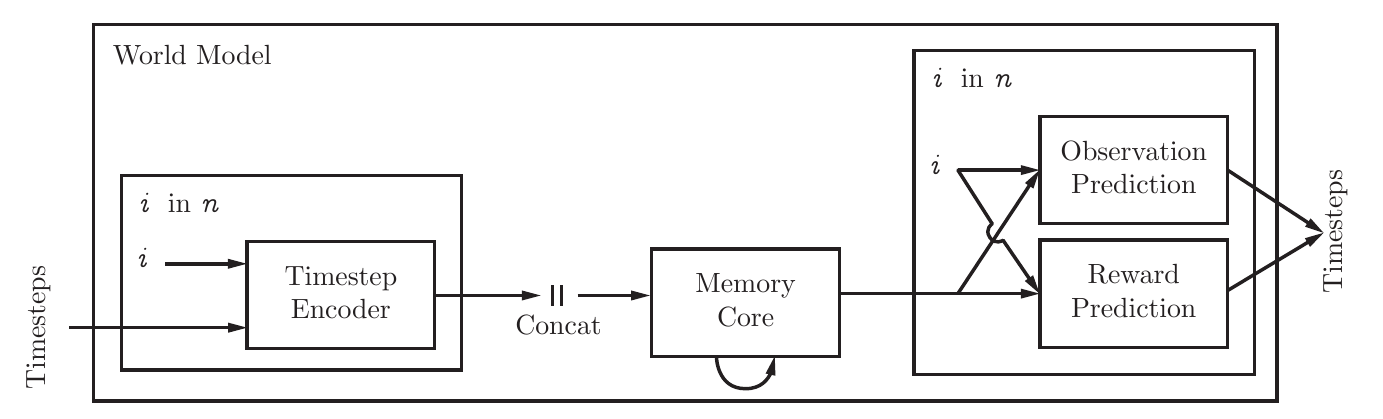}
    \caption{World Model Architecture.}
    \label{fig:world-implementation}
\end{figure}

The high-level architecture of the world model is illustrated in Figure~\ref{fig:world-implementation}.
The world model is composed of several modules that are quite similar to the policy:
\begin{itemize}    
    \item\emph{Timestep Encoder}: 
    Processes all of the current observation's information into a single embedding vector. 
    The timestep includes all new observational data that the agent gains at the current point in time.
    Different from the agent's timestep encoder, this encoder also receives the ID that corresponds with the timestep.
    
    \item\emph{Memory Core}: 
    The component of the agent that maintains and update's the agent's memory.
    Different from the agent's timestep encoder, this memory core receives the representation of each player's timestep concatenated.

    \item\emph{Observation Prediction (Head)}:
    Predicts the successor observation for each player.
    As all games considered in this work are gridworld games, the predicted observation is a classification task for each future grid cell (that are within the respective player's observation window).
    
    \item\emph{Reward Prediction (Head)}:
    Predicts the reward received for each player.
    Rewards are treated as categorical values.
\end{itemize}
Note, that the timestep encoder, observation prediction head, and reward prediction head each use the same parameters across each player.
Similar to the agent, all components are simultaneously trained and their joint parameters are referred to as $\theta^w\in\Theta^w$.
Both observation and reward losses are optimized with a cross entropy objective, and averaged across players. 
The total world model loss is as follows:
\begin{equation*}
    \mathcal{L}_w = \lambda_{\textup{observation}}\cdot\mathcal{L}_{\textup{observation}} + \lambda_{\textup{reward}}\cdot\mathcal{L}_{\textup{reward}}.
\end{equation*}
The implementation of each component is as follows:
\begin{itemize}    
    \item\emph{Timestep Encoder}: 
    The same as the agent's timestep encoder, but the player's ID is also provided alongside the action into the second neural network.
    
    \item\emph{Memory Core}: 
    Identical to the agent.

    \item\emph{Observation Prediction (Head)}:
    The observation prediction is based on the memory core's output and a one-hot ID of the predicted player's ID. 
    These inputs are concatenated and fed into an transposed version of the timestep encoder.
    
    \item\emph{Reward Prediction (Head)}:
    A linear layer of size one.
    For Harvest:~Categorical this output is handled as a discrete prediction; whereas, it is continuous for the other games.
\end{itemize}

A world model is trained for \num{1250000} updates.
Each example in the mini-batch is a sequence of \num{20} transitions, where the first \num{5} timesteps are used to burn-in the memory.
Burn-in does not occur for examples where the first \num{5} transitions are at the beginning of the episode.
Moreover, sequences are added into the replay buffer at a period of \num{14} so that all timesteps show up as prediction targets.

The world model is trained using action-conditioned scheduled sampling (Appendix~\ref{alg:action-biased-scheduled-sampling}, Algorithm~\ref{alg:action-biased-scheduled-sampling}).
The schedule $\epsilon$ follows the following schedule:
\begin{equation*}
    \epsilon (t) = \begin{cases}
        1.0 &t<250000 \\ 
        \frac{4}{3} - \frac{t}{750000} &250000\leq t\leq1000000\\
        0.0 &t>1000000. \\
    \end{cases}
\end{equation*}
This schedule starts out training as a variation of teacher forcing~\cite{williams89teacherforcing}, and slowly transitions to fully auto-regressive.
Additional hyperparameters are specified in Table~\ref{tab:world_hyperparameters}.

\begin{table}[htbp]
    \centering
    \caption{%
        World model hyperparameters per game.
    }      
    \label{tab:world_hyperparameters}
    \begin{tabular}{lrrr}
        \toprule
         Hyperparameter &Harvest:~Categorical &Harvest:~RGB &Running with Scissors \\
         \midrule
         $\lambda_{\textup{observation}}$ &1.0 &1.0 &1.0 \\
         $\lambda_{\textup{reward}}$ &10.0 &0.01 &0.01 \\
         Optimizer &Adam~\cite{kimga14adam} &Adam~\cite{kimga14adam} &Adam~\cite{kimga14adam} \\
         Learning Rate &\num{3e-4} &\num{3e-4} &\num{3e-4} \\
         Max Grad Norm &\num{10.0} &\num{10.0} &\num{10.0} \\
         Batch Size &\num{32} &\num{24} &\num{24} \\
         \bottomrule
    \end{tabular}
\end{table}

%% file: appendices/methods/strategic_diversity.tex
Learning a general world model assumes that the transitions are drawn from the space of all possible transitions.
This is typically not tractable, but instead draws are taken from a dataset generated from play of a \emph{behavioral strategy} $\bm\sigma$.
And the performance of the world model is measured under a \emph{target strategy} $\bm\sigma^*$, instead of all possible strategies $\bm\Sigma$.
Differences between $\bm\sigma$ and $\bm\sigma^*$ present challenges in learning an effective world model.

We call the probability of drawing a state-action pair $\bm{s},~\bm{a}$ under some strategy its \emph{reach probability} $\eta^{\bm{\hat\sigma}}$ under joint strategy $\bm{\hat\sigma}$.
From this, we define \emph{strategic diversity} as the distribution induced from reach probabilities.
These terms allow us to observe two challenges for learning world models.

First, the diversity of the behavioral strategy \emph{cover} the target strategy's diversity:
\begin{equation}
    \eta^{\bm\sigma^*}(\bm{s},\bm{a})\to\eta^{\bm{\sigma}}(\bm{s},\bm{a}).
\end{equation}
Otherwise, transitions will be absent from the training data.
As an aside, it is possible to construct a weaker claim for coverage.
This is done through making additional assumptions about the generalization capacity of a world model across transitions.
For example, if transitions are drawn from two discrete latent variables, unseen combinations of these variables may be generalized if the individual values are known.
However, generalization cannot be generally guaranteed, so we consider coverage.

The second challenge is that the \emph{closer} the diversities are, the more accurate the learning objective will be.
In other words, we want 
\begin{equation}
    \eta^{\bm\sigma^*}(\bm{s},\bm{a})\approx\eta^{\bm{\sigma}}(\bm{s},\bm{a}).
\end{equation}
If closeness is not ensured, crucial dynamics knowledge may not be learned as the learning signal is dominated from unimportant transitions.
An example of the issue of closeness can be seen in the ``noisy TV problem,''~\cite{burda2018exploration}.
This exploration problem poses that novelty-seeking agents may be stuck forever watching the ever new TV static, and not experiencing practical novelty.
In the same vein, if a world model is trained almost entirely on ``noisy TV''-like experiences, and as a rarely on the few salient experiences, it may never learn.
Therefore, we should strive to correct the distribution of experiences to be informed by a target strategy.

By design, empirical-game building algorithms offer a means to construct the target world model objective.
These algorithms require the specification of a solution concept that serves the dual roll as the target strategy for a world model.
Then through an iterative process, the empirical-game constructs strategies that progressively approach the target.
In turn, generating transitions that match the target world model objective.

\begin{claim}
Dyna-PSRO produces a correct world-model objective $\eta^{\bm\sigma^*}$ with a best-response oracle and a correct empirical game for a game with a unique Nash Equilibrium $\bm\sigma^*$.
\end{claim} 

\begin{proof}
Following McMahan et al.~\cite{mcmahan03do}, the Double Oracle algorithm will converge to a NE in the limit of enumerating the full strategy space.
Let $\bm\sigma^0, \bm\sigma^1, \ldots, \bm\sigma^e$ be the solutions discovered for each epoch, ending at epoch $e$.
Then a dataset composed of experiences generated by the current empirical game solution evolves as follows:
\begin{equation}
    \eta^{\bm\sigma^0} \to \eta^{\bm\sigma^1} \to \ldots \to \eta^{\bm\sigma^e} = \eta^{\bm\sigma^*}.
\end{equation}
% Therefore, the final objective is correctly:
% \begin{equation}
%     \min_{\theta^w\in\Theta^w} \mathbb{E}_{(\bm{o}, \bm{a}, \bm{o}', \bm{r})\sim \eta^{\bm\sigma^*}} \left[ \mathcal{L}(w(\bm{o}, \bm{a} \mid \theta^w), (\bm{o}', \bm{r})) \right].
% \end{equation}
\end{proof}

The previous claim contains two strong assumptions: an exact best-response oracle and error-less empirical game.
These assumptions must be made, because PSRO is parameterized by its choice of response oracle and empirical game model; therefore, PSRO's convergence must be proven for each choice.
Theoretically PSRO has been shown to converge to an $\epsilon$-NE, where $\epsilon$ depends on the empirical game's modelling error,  to a corresponding NE in the true game~\cite{tuyls20egtabounds, vorobeychik10}.
Therefore, in practice Dyna-PSRO produces $\eta^{\bm\sigma^e} \approx \eta^{\bm\sigma^*}$, which supports the weaker claim that Dyna-PSRO generally improves the quality of a world model.

It is also worth noting the connections between this analysis and MARL regimes that seek to find any solution the game.
In these regimes, the priority is finding \emph{any} performant strategy.
This matches the approach taken by the majority of studies in MARL falling under paradigms such as Independent RL or Self-Play.
Therefore, their target distribution is the best-response to the previous strategy $\eta^{\textup{BR}(\bm\sigma^{i-1})}$ and changes in tandem with the strategies.
When no best-response can be found, then the current strategy matches the solution and the dataset correspondingly reflects this.

%% file: appendices/methods/dyna_psro.tex
The Dyna-PSRO builds upon PSRO (Algorithm~\ref{alg:psro}) by including the co-learning of a world model.
The high-level pseudocode of Dyna-PSRO is provided in Algorithm~\ref{alg:dyna-psro} an a high-level application architecture diagram is depicted in Figure~\ref{fig:dyna-psro-processes}.
There are three main co-routines of Dyna-PSRO: response computation, world-model learning, and empirical-game simulation. 
The details of each routine are first provided; then, how the routines interact with each other is explained.

\begin{figure}[!ht]
    \centering
    \includegraphics{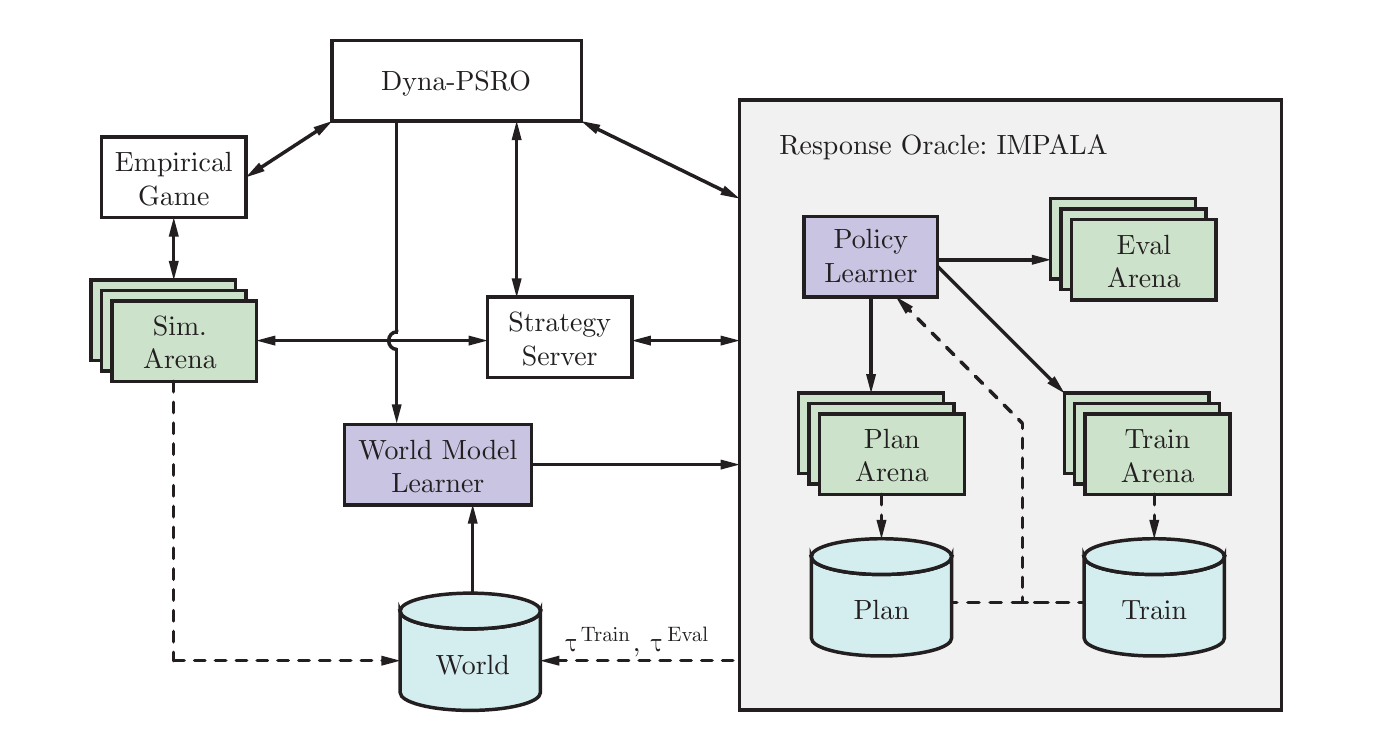}
    \caption{%
        Overview of the major Dyna-PSRO processes. 
    }
    \label{fig:dyna-psro-processes}
\end{figure}

\subsubsection{Empirical Game}
The empirical game routine is responsible for maintaining the empirical game, including simulating new payoffs and game reasoning.
New profiles are sent to \emph{simulation (sim.) arenas} for payoff estimation in parallel.
Once all profiles are estimated, the game is solved, and the solution is based to the main Dyna-PSRO process.
In the experiments in this work, the chosen solution is Nash Equilibrium, and it is solved through the linear complementarity~\cite{eaves71lcp} algorithm that is implemented by Gambit~\cite{gambit}.

\subsubsection{World Model}
The world model routine is responsible for training the world model and serving its parameters. 
This routine's pseudocode is provided in Algorithm~\ref{alg:world-model-learner}, and follows mostly the same method details as the strategic diversity experiment. 
The difference is that instead of there being a precomputed fixed dataset, the world model is now trained over a dynamic dataset.
The dataset is represented by a replay buffer that is populated from: (1) trajectories from the simulation arena used for expanding the empirical game, and (2) trajectories from the training and evaluation arenas from the response calculation. 
Notably, all of this data must be generated in the standard PSRO procedure, so it collected with no additional cost.
The world-model learner samples and evicts data randomly from this buffer.

% \begin{wrapfigure}[11]{R}{0.46\textwidth}
\begin{algorithm}[H]
\caption{World Model Learner}
\label{alg:world-model-learner}
\KwIn{World model $w$ and data buffer $\mathcal{B}^w$}
\KwIn{$n$ no. of updates (default: $\infty)$.}

\For{$i\in[[n]]$}{
    Train $w$ over $\tau\sim\mathcal{B}^w$\;
}

\KwOut{$w$}
\end{algorithm}
% \end{wrapfi ngure}

\subsubsection{Response Oracle}
The response oracle uses the IMPALA~\cite{espeholt18impala} algorithm to compute an approximate best-response to the opponent's strategy according the the current empirical game.
IMPALA uses several processes that generate experiences for the agent to train on. 
These process are referred to in this work as arenas.
The \emph{train arenas} generate real experiences, and the \emph{plan arenas} generate planned experiences.
If the learner is using decision-time planning they will only use it in the train arenas.
A third set of arenas called \emph{eval arenas} periodically evaluate the performance of the greedy policy and record additional metrics.
The arenas attempt to synchronize all parameters at the start of each episode.

The policy learner runs for a fixed number of updates, querying the datastores for experiences to learn from. 
The specifics of how each policy learns is described in Appendix~\ref{app:policy}.

\subsubsection{Runtime Procedure}
A sketch of the respective processes runtime is shown in Figure~\ref{fig:dyna-psro-runtime}
As in PSRO, the main empirical-game building loop iterates between response computation and empirical-game simulation.

\begin{figure}[!ht]
    \centering
    \includegraphics{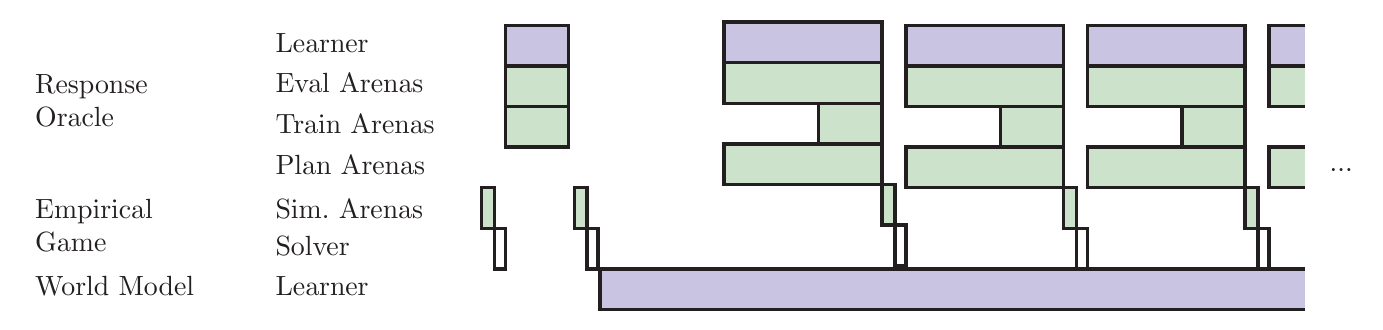}
    \caption{%
        Example Dyna-PSRO runtime.
        Planning is set to occur after the first epoch.
        Each players' response oracle runs in parallel.
    }
    \label{fig:dyna-psro-runtime}
\end{figure}

The runtime is defined by a parameter specifying on which PSRO epoch to begin planning.
Before that epoch, the response oracles do not use planning at all, because the world models are untrained.
All of these policies therefore are trained exclusively on real experiences just like standard PSRO.
However, these experiences are also being used to populate the world model's replay buffer.
Once the first planning epoch has arrived, computing responses is temporarily paused.
The world model is then given a set number of updates to warm-start its parameters, before being used in response calculation.
Once the world model's warm-start phase is over, all process proceed concurrently.

Throughout this work planning begins on the second epoch.
The world models is given 1 million updates of warm starting.

\begin{minipage}[h]{0.55\textwidth}
\begin{algorithm}[H]
\caption{Policy-Space  Response Oracles~\cite{lanctot17psro}}
\label{alg:psro}
\KwIn{Initial strategy sets for all players $\bm\Pi^0$} 

Simulate utilities $\hat{U}^{\bm\Pi^0}$ for each joint $\bm\pi^0\in\bm\Pi^0$\;
Initialize solution $\sigma_i^{*,0} = \text{Uniform}(\Pi_i^{0})$\;

\While{epoch e in $\{1, 2,\dotsc \}$}{
    \For{player $i\in [[n]]$}{
        \tcp*[h]{Algorithm~\ref{alg:response-oracle}.}
        $\pi_i^e,\_ = \textup{response\_oracle}(\sigma_{-i}^{*,e-1})$\;
        $\Pi^e_i = \Pi^{e-1}_i \cup \left\{ \pi^e_i \right\}$\;
    }
    Simulate missing entries in $\hat{U}^{\Pi^e}$ from $\bm\Pi^e$\;
    Compute a solution $\sigma^{*,e}$ from $\hat{\Gamma}^e$\;
}

\KwOut{Current solution $\sigma_i^{*,e}$ for player $i$}
\end{algorithm}
\end{minipage}
\begin{minipage}[h]{0.45\textwidth}
\begin{algorithm}[H]
\caption{Response Oracle}
\label{alg:response-oracle}
\KwIn{Coplayer strategy profile $\sigma_{-i}$}
\KwIn{Num updates $k$}
$\pi_i\gets\theta^\pi$\;
$\mathcal{B}\gets \left\{\right\}$ \tcp*{Replay Buffer.}
\For{many async episodes}{
    $\pi_{-i}\sim\sigma_{-i}$\;
    $\mathcal{B} = \mathcal{B} \cup \left\{ \tau \sim (\pi_i,\pi_{-i}) \right\}$\;
}
\For{$i\in[[k]]$}{
    Train $\pi_i$ over $\tau\sim\mathcal{B}$\;
}
\KwOut{$\pi_i,\mathcal{B}$}
\end{algorithm}
\end{minipage}

\begin{algorithm}
\caption{Dyna-PSRO}
\label{alg:dyna-psro}
\KwIn{Initial strategy sets for all players $\bm\Pi^0$}
\KwIn{No. of world model head-start updates $n_w$}
\KwIn{Epoch to begin planning $e^\textup{plan}$}

Simulate utilities $\hat{U}^{\bm\Pi^0}$ for each joint $\bm\pi^0\in\bm\Pi^0$\;
Initialize solution $\sigma_i^{*,0} = \text{Uniform}(\Pi_i^{0})$\;
$w\gets\theta^w$\;
$\mathcal{B}^w\gets\left\{\right\}$ \tcp*{World Model's Replay Buffer.} \;

\While{epoch e in $\{1, 2,\dotsc \}$}{
    \For{player $i\in [[n]]$}{
        \If{$e > e^\textup{plan}$}{
            $\pi_i^e ,\tau = \textup{async}(\textup{planner\_oracle}(\sigma_{-i}^{*,e-1},w))$ \tcp*{Algorithm~\ref{alg:planner-oracle}.}
        }
        \Else{
            $\pi_i^e ,\tau = \textup{async}(\textup{response\_oracle}(\sigma_{-i}^{*,e-1}))$ \tcp*{Algorithm~\ref{alg:response-oracle}.}     
        }
        $\mathcal{B}^w = \mathcal{B}^w \cup \left\{ \tau \right\}$\;
        $\Pi^e_i = \Pi^{e-1}_i \cup \left\{ \pi^e_i \right\}$\;
    }
    Wait on all futures $\bm\pi^e,\tau$\;\;
    
    Simulate missing entries in $\hat{U}^{\bm\Pi^e}$ from $\bm\Pi^e$\;
    Add $\tau$ from simulation to $\mathcal{B}^w$\;
    Compute a solution $\bm\sigma^{*,e}$ from $\hat{\Gamma}^e$\;\;

    \If{e == 1}{
        $w = \textup{world\_model\_learner}(w,n_w)$ \tcp*{Algorithm~\ref{alg:world-model-learner}.}
        $w = \textup{async}(\textup{world\_model\_learner}(w))$ \tcp*{Parameters periodically sync.}
    }
}

\KwOut{Current solution $\sigma_i^{*,e}$ for player $i$}
\end{algorithm}

\begin{algorithm}[H]
\caption{Planner Oracle}
\label{alg:planner-oracle}
\KwIn{Coplayer strategy profile $\sigma_{-i}$}
\KwIn{World model $w$, real game dynamics $p$}
\KwIn{Warm-start background planning updates $n^\textup{BG:WS}$}
\KwIn{Training updates $n$}
\KwIn{Concurrent background planning fraction $f^\textup{BG:C}$}
$\pi_i\gets\theta^\pi$\;
$\mathcal{B}^\textup{plan}\gets \left\{\right\}$ \tcp*{Replay Buffer with planned experience.}
$\mathcal{B}^\textup{train}\gets \left\{\right\}$ \tcp*{Replay Buffer with real experience.}
\;
\tcp*[h]{Asynchronously generate data on arenas.}\;
\For{many async episodes}{
    $\pi_{-i}\sim\sigma_{-i}$\;
    $\mathcal{B}^\textup{plan} = \mathcal{B}^\textup{plan} \cup \left\{ \tau \sim (\pi_i,\pi_{-i},w) \right\}$\;
}
\For{many async episodes}{
    $\pi_{-i}\sim\sigma_{-i}$\;
    $\mathcal{B}^\textup{train} = \mathcal{B}^\textup{train} \cup \left\{ \tau \sim (\pi_i,\pi_{-i},p) \right\}$\;
}
\;
\tcp*[h]{Train the response policy.}\;
\For{$i\in[[n^\textup{BG:WS}]]$}{
    Train $\pi_i$ over $\tau\sim \mathcal{B}^\textup{plan}$\;
}
\For{$i\in[[n]]$}{
    Train $\pi_i$ over $\tau\sim \left\{(1.0 - f^\textup{BG:C})\cdot\mathcal{B}^\textup{train}\right\} \cup \left\{f^\textup{BG:C}\cdot\mathcal{B}^\textup{plan}\right\}$\;
}
\KwOut{$\pi_i,\mathcal{B}$}
\end{algorithm}

%% file: appendices/methods/regret.tex
Combined-game regret is an approximate measure of regret that all available estimates to approximate the regret within the true game.
Intuitively, combined-game regret is the regret of a strategy with respect to all discovered policies.
When comparing empirical-game building algorithms this is formalized as follows:
\begin{equation}
    \text{SumRegret}(\bm{\sigma}, \overline{\bm{\Pi}})=\sum_{i\in n}\max_{\pi_i\in\overline\Pi_i}\hat{U}_i(\pi_i, \sigma_{-i}) - \hat{U}_i(\sigma_i, \sigma_{-i}), \qquad \overline{\bm{\Pi}}_i\equiv\bigcup_\textup{method}\hat{\bm{\Pi}}_i^\textup{method},
\end{equation}
where $\hat{\Pi}$ is the restricted strategy set from one of the algorithms.

\begin{figure}[!ht]
    \centering
    \includegraphics{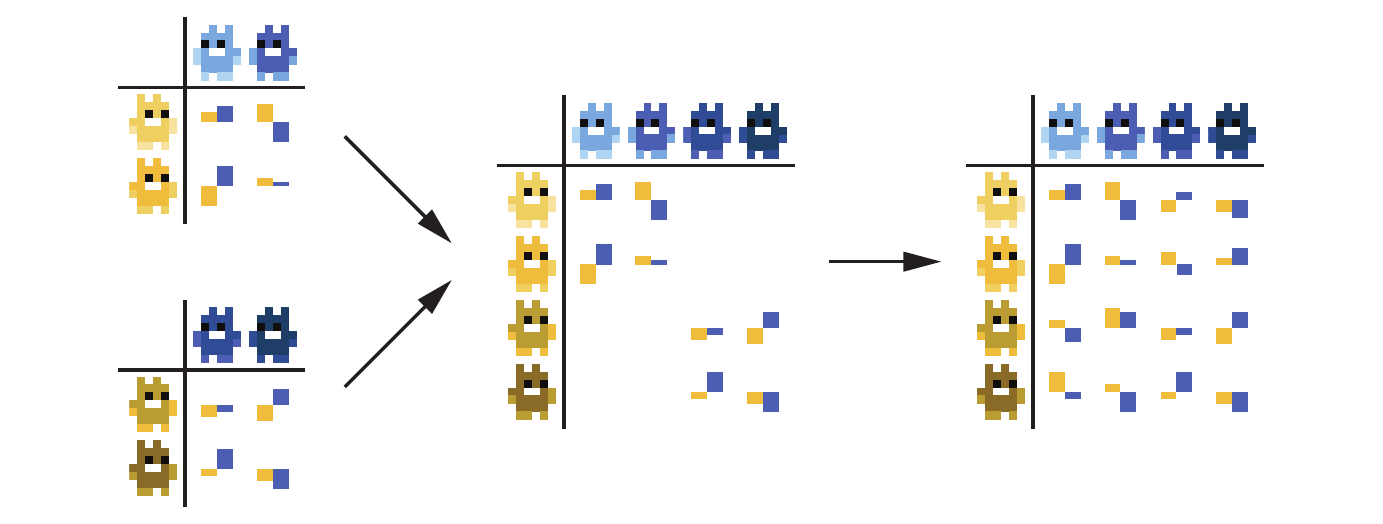}
    \caption{%
        Combined-game construction.
        Left: Constituent empirical games.
        Middle: Combination of the strategy sets and payoff functions.
        Right: Completion of the empirical game by estimating new strategy profile payoffs. 
    }
    \label{fig:combined-game}
\end{figure}

The process of constructing a combined-game is illustrated in Figure~\ref{fig:combined-game}.
Where, the strategy sets (depicted by the toons) across methods are first combined.
The new \emph{combined game} that results from this can be initialized with the payoff estimates from the constituent empirical games.
Unestimated payoffs must then be simulated for the new strategy profiles.
Then the complete combined game can be used to compute the combined-game regret from the solutions computed throughout the empirical-game building algorithms.

%% file: appendices/games/harvest_categorical.tex
In Harvest, players move around an orchard picking apples. 
The challenging commons element is that apple regrowth rate is proportional to nearby apples, so that socially optimum behavior would entail managed harvesting. 
Self-interested agents capture only part of the benefit of optimal growth, thus non-cooperative equilibria tend to exhibit collective over-harvesting. 
The game has established roots in human-behavioral studies~\cite{janssen10commons} and in agent-based modeling of emergent behavior~\cite{perolat17commons, leibo17, leibo2021meltingpot}.

\begin{figure}[htbp]
    \centering
    \includegraphics{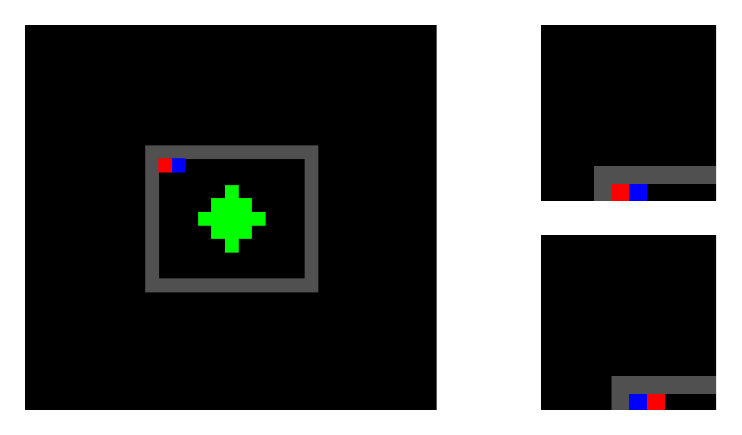}
    \caption{%
        Harvest:~Categorical.
        Left: game state.
        Right: player observations.
    }
    \label{fig:gathering}
\end{figure}

For our initial experiments, we use a symmetric two-player version of the game, where in-game entities are represented categorically~\cite{hcai19gathering}.
This categorical representation facilitates faster experimentation and simplifies the interpretation of results. 
Figure~\ref{fig:gathering} depicts the game state and player observations.
Each player has a $10\times10$ viewbox within their field of vision. 
The cells of the grid world can be occupied by either agent shown in red and blue, the apples shown in green, or a wall in gray.
The possible actions include moving in the four cardinal directions, rotating either way, tagging, or remaining idle.
A successful tag temporarily removes the other player from the game, but can only be done to other nearby players.
Players receive a reward of \num{1} for each apple picked.
Episodes are limited to \num{100} timesteps.

%% file: appendices/games/harvest_rgb.tex
Harvest:~RGB is a different implementation of the Harvest game introduced by Harvest:~Categorical (Appendix~\ref{app:games:harvest-categorical}.
Harvest:~RGB is exactly the harvest implementation from MeltingPot~\cite{leibo2021meltingpot} with the same orchard map.
A rendering of the game state and observations is shown in Figure~\ref{fig:harvest}.
The main difference between the Harvest versions is that the observations are $88\times 88\times 3$ images of the $11\times 11$ viewbox in front of them.
There are also minor differences in the implementation of tagging and apple respawn mechanism.
Episodes play for \num{1000} timesteps.

\begin{figure}[htbp]
    \centering
    \includegraphics{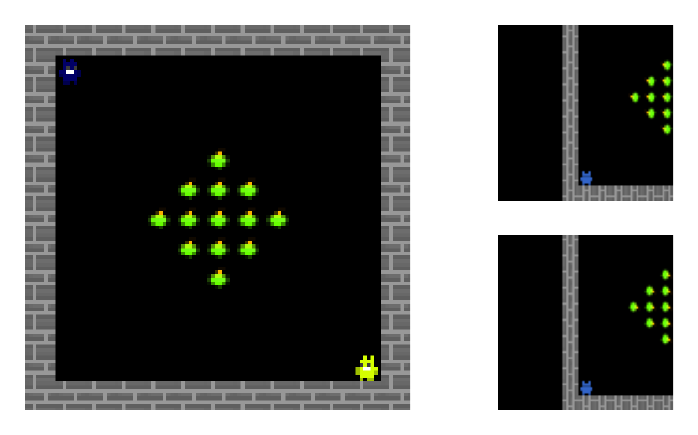}
    \caption{%
        Harvest:~RGB.
        Left: game state.
        Right: player observations.
    }
    \label{fig:harvest}
\end{figure}

%% file: appendices/games/running_with_scissors.tex
Running With Scissors (RWS) is a temporally extended version of rock-paper-scissors (RPS).
In it, players collect rock, paper, and scissor items into their inventory.
At any point the player has the option to tag their opponent if they're nearby.
Then they play RPS corresponding to the distribution of items in their inventories.
The agents have the same action space as in the previous games.
The observation space is $40\times 40\times 3$ image-based viewbox in fromt of them corresponding to a $6\times 6$ grid around them.
A portion of items are placed within the game deterministically, the rest are randomly sampled before play.
If neither player tags each other before \num{1000} timesteps, the players are forced into playing RPS. 

\begin{figure}[htbp]
    \centering
    \includegraphics{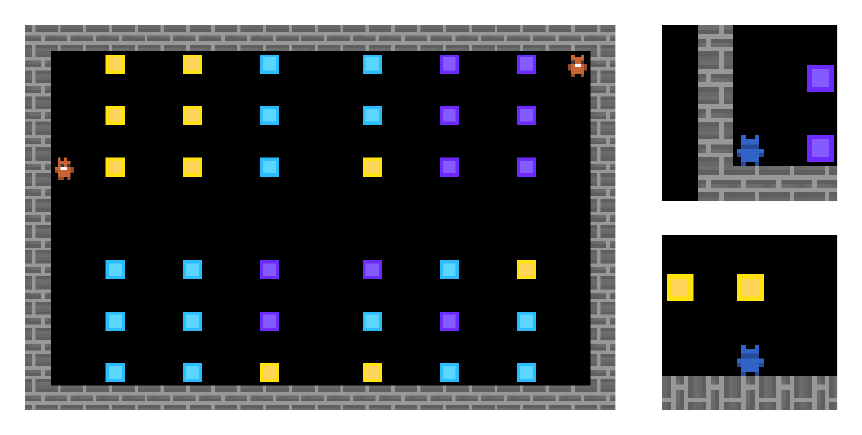}
    \caption{%
        Running With Scissors.
        Left: game state.
        Right: player observations.        
    }
    \label{fig:running-with-scissors}
\end{figure}

%% file: appendices/results/diversity.tex
Figure~\ref{fig:gathering_recall} displays the recall results that correspond to the accuracies portrayed in Figure~\ref{fig:world-model-comparison-accuracy}.
See Section~\ref{sec:diversity} for a discussion of these results.

\begin{figure}[h]
    \centering
    \includegraphics{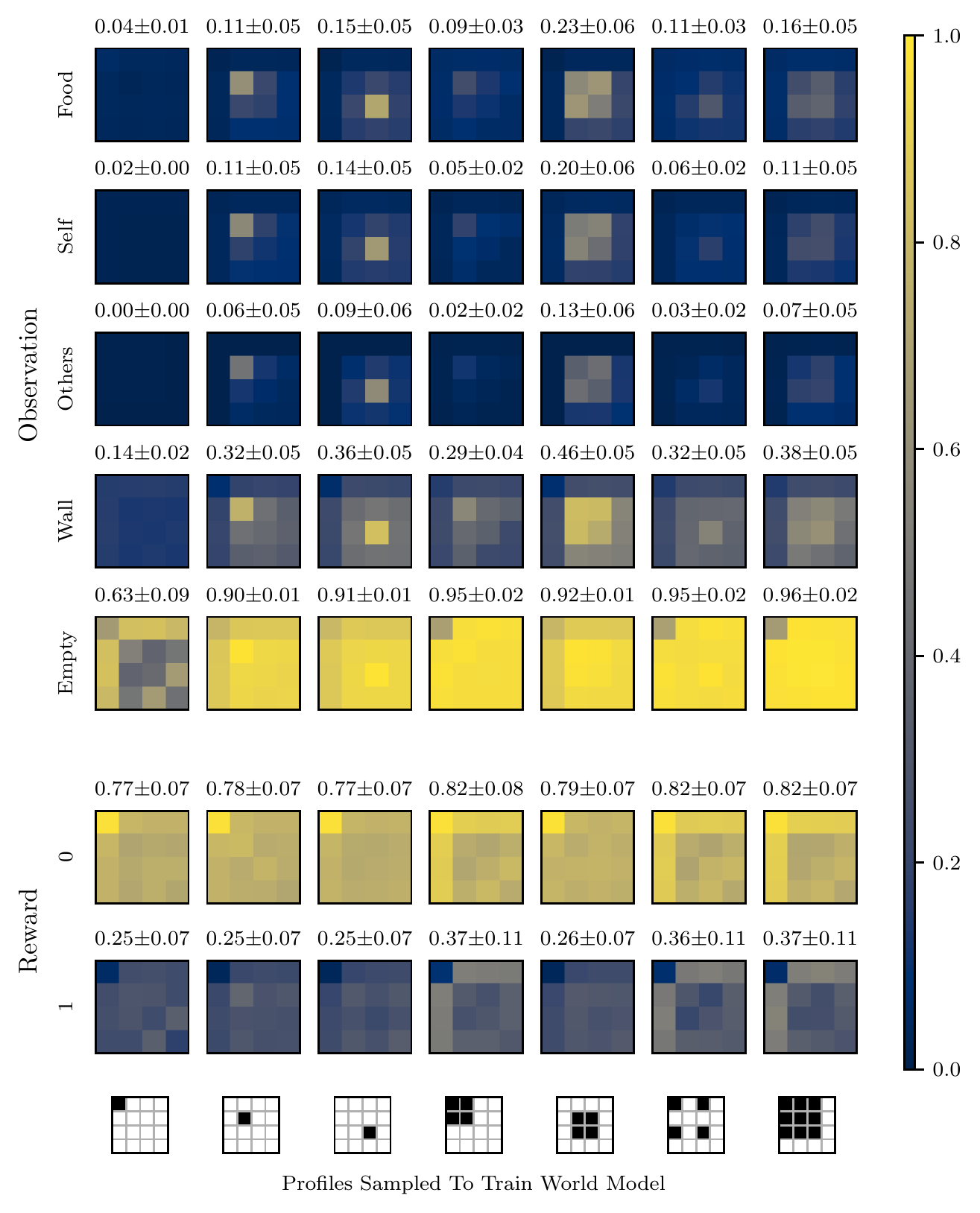}
    \caption{%
        World model recall on Harvest:~Categorical.
    }
    \label{fig:gathering_recall}
\end{figure}

%% file: appendices/results/background_planning.tex
Figure~\ref{fig:background-planning-bad} shows the results of repeating the background planning experiment with world model \wmA\!\!.
Besides changing the world model, the methodology is consistent with that described in Section~\ref{sec:background-planning}.
This result shows the planner achieving results comparable to the baseline method.
Supporting the adoption of planning, as it tends to not negatively impact the learning process.

\begin{figure}[h]
    \centering
    \includegraphics{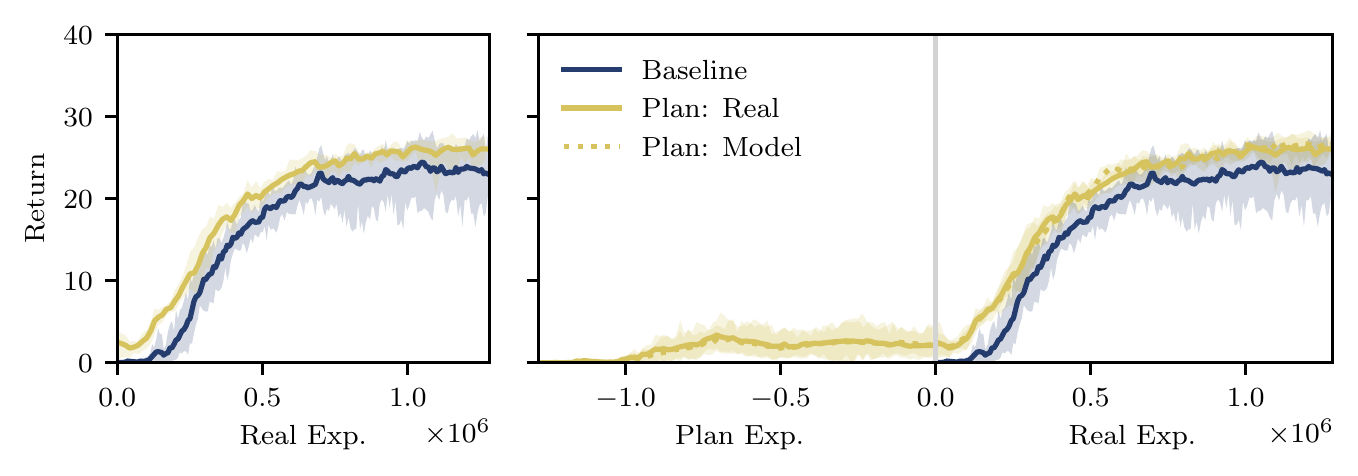}
    \caption{%
        Effects of background planning on response computation using world model \wmA\!\!.\\
        (\num{5}~seeds, with \SI{95}{\percent} bootstrapped CI).        
    }
    \label{fig:background-planning-bad}
\end{figure}

%% file: appendices/results/decision_time_planning.tex
Figure~\ref{fig:decision-time-planning-bad} shows the results of repeating the decision-time planning experiment with world model \wmA\!\!.
Besides changing the world model, the methodology is consistent with that described in Section~\ref{sec:decision-time-planning}.
This result further exemplifies the trend shown in Figure~\ref{fig:decision-time-planning}, where the planners that did not use BG:~W failed to learn an effective policy.
The planner that used BG:~W achieved performance comparable to the baseline.
Finally, the planner that used both BG:~W and BG:~C achieves the strongest performance at \num{33.07\pm6.76}.
These results support the benefit of BG:~W when using DT, and that effective planning performs as least as good as the baseline.

\begin{figure}[h]
    \centering
    \includegraphics{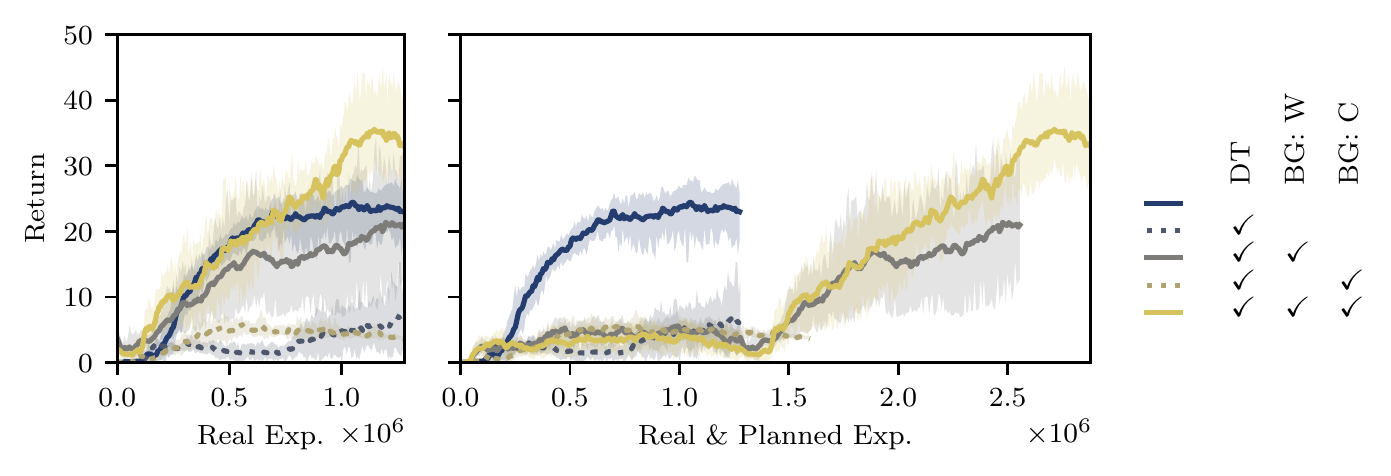}
    \caption{%
        Effects of decision-time planning on response computation using world model \wmA\!\!.
        (\num{5}~seeds, with \SI{95}{\percent} bootstrapped CI).
    }
    \label{fig:decision-time-planning-bad}
    % baseline      \num{22.54\pm1.77}
    % search        \num{6.29\pm4.10}
    % search-warm   \num{20.98\pm8.20}
    % search-bg     \num{3.64\pm0.26
    % search-both   \num{33.07\pm6.76}
\end{figure}

%% file: appendices/results/world_model_enfg.tex
In this section, we verify the need for separate models.

First, consider if an empirical game can substitute for a world model.
The majority of previous work on empirical games represents the model in the normal form. 
This representation abstracts away any notion of dynamics within an episode into a choice in policy and the resulting payoff. 
Since empirical games currently lack dynamics information completely, this supports the choice of separate models.
This is not without any exceptions.
If the original game is one-shot and stateless (i.e., an episode is played through a single action), then a normal-form empirical game is exactly a world model.

Now, consider if a world model can substitute for an empirical game.
World models predict successor states and rewards; and thus, can rollout planned trajectories to estimate payoffs. 
Note, that rolling out a trajectory a trajectory with a world model is an auto-regressive prediction that tends to result in compounding errors~\cite{talvitie14, holland18}.
Despite this, it is plausible that a world model can substitute as a high-fidelity empirical game.

\begin{figure}[htbp]
    \centering
    \includegraphics{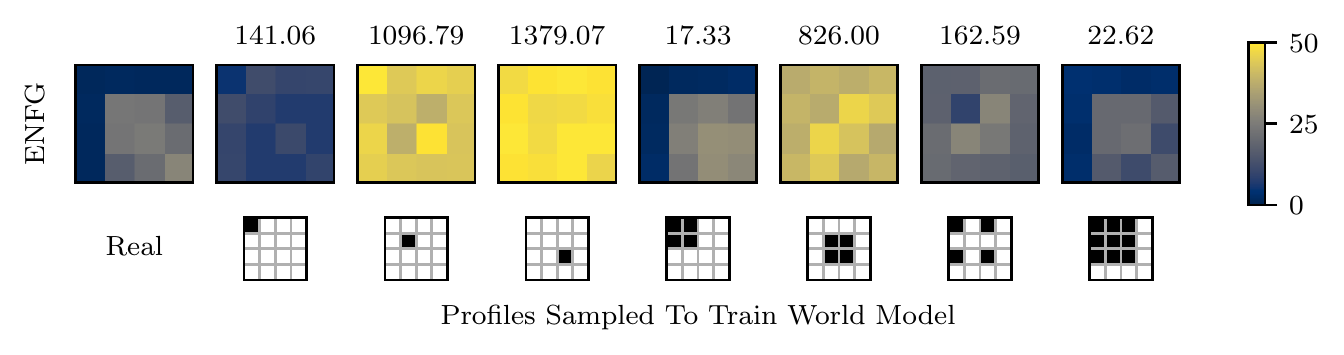}
    \caption{%
        Empirical normal-form games (ENFG) estimated by world model rollouts. 
        The title of each plot is its L2 distance with the real ENFG. 
    }
    \label{fig:world-model-enfg}    
\end{figure}

Figure~\ref{fig:world-model-enfg} compares an empirical game estimated from real game payouts empirical games estimated with payouts predicted by a world model.
In this experiment, the world models are the same that were used in Section~\ref{sec:diversity}.
In general, the empirical games estimated by world models have large errors (L2 \num{>100}), with several having exceptionally large errors (L2 \num{>1000}).
These result suggest that this direction may be possible with future algorithmic improvements; however, currently, the prediction errors are too large to substitute empirical games with world models.
Especially in games with long time horizons.